\documentclass[10pt,conference,compsocconf,letterpaper]{IEEEtran}

\usepackage{amssymb,url,amsmath, relsize, accents}
\usepackage{mathrsfs, enumerate}
\usepackage{graphicx}
\usepackage{multicol}
\usepackage{galois}
\usepackage{amsthm}
\usepackage[T1]{fontenc}
\usepackage[utf8]{inputenc}
\usepackage{authblk}
\usepackage{subcaption}

\def\Pr{{\mathbb P}}

\def\eqdef{\triangleq}

\newtheorem{theorem}{Theorem}

\newtheorem{lemma}{Lemma}

\newtheorem{assumption}{Assumption}

\usepackage{amsthm} 

\newcounter{example}[section]

  \newcounter{remark}[section]

 \newcounter{experiment}[section]
\newenvironment{experiment}[1][]{\refstepcounter{experiment}\par\medskip
   \noindent \textbf{Experiment~\theexperiment. #1} \rmfamily}{\medskip}

\begin{document}

\title{Asymptotic Miss Ratio of LRU Caching with Consistent Hashing}
\date{}
\author{Kaiyi Ji, Guocong Quan, Jian Tan\\
Department of Electrical and Computer Engineering\\
The Ohio State University, Columbus, 43210\\
              Email: \{ji.367, quan.72, tan.252\}@osu.edu}

\maketitle

\begin{abstract}
To efficiently scale data caching infrastructure to support emerging big data applications, many caching systems rely on consistent
hashing to group
a large number of servers to form a cooperative cluster.
These servers are organized together according to a random hash function.
They jointly provide a unified but distributed hash table to serve swift and voluminous
data item requests. 
Different from the single least-recently-used (LRU) server that has already been extensively studied, theoretically characterizing a cluster
that consists of multiple LRU servers
remains yet to be explored.
These servers are not simply added together; the random hashing
complicates the behavior.
To this end, we derive the asymptotic miss ratio of data item requests on a LRU cluster
with consistent hashing.
We show that these individual cache spaces on different servers  
can be effectively viewed as if
they could be pooled together 
to form a single virtual LRU cache space parametrized by an appropriate cache size.  This equivalence can be established rigorously
under the condition that the cache sizes of
the individual servers are large. For typical data caching systems this condition is common.  
Our theoretical framework provides a convenient abstraction that can directly apply the results from the simpler
single LRU cache to the more complex LRU cluster with consistent hashing.  
\end{abstract}

\section{Introduction}\label{s:intro} 

With the advent of cloud computing and emergence of big data,  scale-out data caching systems are widely deployed and their
horizontal scalability~\cite{singh2015survey} becomes increasingly important. As an effective solution, consistent hashing~\cite{Karger:1997} 
has been commonly used by key-value caching systems, e.g., Dynamo~\cite{dynamo}, Aerospike~\cite{aerospike}, Memcached~\cite{memcached}, Redis~\cite{redis}. Using consistent hashing,  a large
number of servers are organized together to form a cooperative cluster. These servers jointly provide a 
unified but distributed hash table to serve swift and voluminous
 data item requests.     
Once a data item is hashed to one of the hosting servers, most key-value caching systems use
the least-recently-used (LRU) caching algorithm, or its variations, e.g.,  LRU Clock~\cite{AndrewOS}, to
decide which data items should be kept in its own individual cache space. 
These data caching systems play a critical role in optimizing the way information is delivered in Web services. 


With consistent hashing, the total amount of cache spaces in the cluster
can be easily expanded (scaled horizontally) through the addition of new cache servers. 
However, these individual cache spaces on different servers are not simply added together to achieve an overall request miss ratio.  
Although LRU caching on a single server has already been extensively studied,
theoretically characterizing the miss ratio of a LRU cluster 
organized by consistent hashing
still remains an unexplored problem. 
One difficulty in analysis is that the data item requests are shuffled to a large number of servers according to a random hash function. 
This random hashing complicates the system behavior. 
Due to the fundamental role and predominant usage in practice, LRU caching with consistent hashing merits a deep investigation. 

Characterizing the cache miss behavior of a cluster in such a complex setting not only helps 
resource planning but also improves the way a cache cluster is organized.   To this end, we derive the asymptotic miss ratio of a LRU caching cluster with consistent hashing under the independent reference model (IRM)~\cite{vanichpun2004output}.  Interestingly, these individual cache spaces on different servers, though isolated physically but logically connected through a hash function,    
can be effectively viewed as if 
they could be pooled together to form a single virtual LRU cache space.  Interestingly, this virtual LRU cache has a equivalent cache 
size determined by the distribution of the random hash function.  This equivalence can be established when the cache sizes of 
the individual servers are large. 
Our result provides a convenient abstraction 
that can rigorously relate the more complex LRU caching with consistent hashing to the relatively simpler single LRU cache.  
Based on this abstraction, many known results on  
a single LRU cache can be directly translated to a LRU cluster with consistent hashing.  
Specifically, we prove a characteristic time approximation, previously established for a single LRU cache,  for a LRU cluster.  
Notably, the characteristic time approximation has the same form for almost all of the random hash functions. 
This result is not straightforward in view that the miss ratio of each 
of the server is a random variable conditional on the random hash function.  
However, the overall conditional asymptotic miss ratio of the cluster is always the same almost surely,  depending on the probability distribution of the random hash function.
Due to this equivalence, we comment that the engineering implications discussed in~\cite{jiantanSig} for a single LRU cache can also be extended to a LRU cluster.

%
%
%
%

\subsection{Background}
To put the analysis on a concrete basis,  we first summarize the important features of consistent hashing and LRU caching. 
\subsubsection{Consistent hashing}
Data items are usually organized in a key
value pair, and the entire data are stored in the whole cluster as a distributed
hash table according to keys. 
\begin{figure}[h]
\vskip -0.16in
\centering
\includegraphics[width=8.0cm]{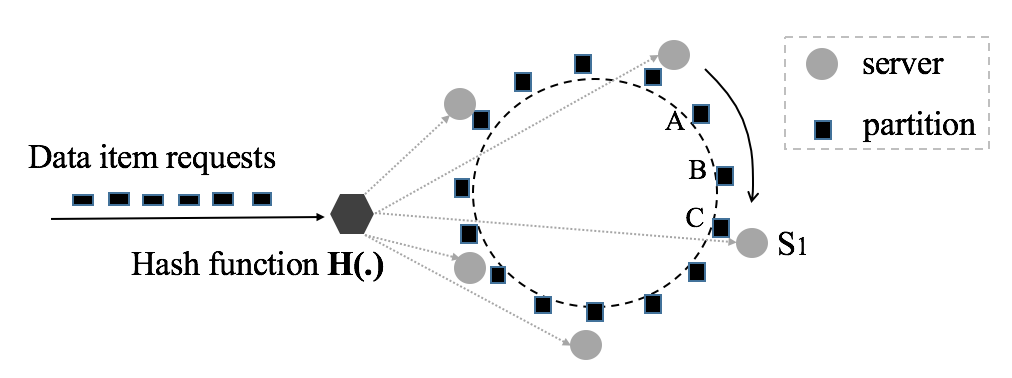}
\caption{Illustration of consistent hashing}
\label{fig:consisHashing}
\vspace{-0.2cm}
\end{figure}
To scale out the system horizontally, each server is maintained independent, e.g., by using consistent hashing~\cite{Karger:1997} to select a unique server for each key.
The basic idea is to first hash the 
data items to a large number of partitions at random.  These partitions form a ring, as illustrated in Fig.~\ref{fig:consisHashing}. The physical servers, with a total number that is much smaller than the number of partitions,  
are also hashed to a subset of the partitions. 
A data item, after being hashed to a partition,  will be stored on the server that is closest in the clockwise direction to its associated partition.  For example, partitions~A, B and C are all stored on server $S_1$ in Fig.\ref{fig:consisHashing}.  
Therefore, when the server locations are fixed, each server hosts a certain number of partitions that are determined by the hash function. 

\subsubsection{LRU caching}
Each data item is hashed to one of the hosting servers. A cache replacement algorithm is needed to manage the cache space 
on an individual server.  Due to the low cost of tracking access history, LRU caching algorithm has been widely used, e.g., for Memcached~\cite{memcached,memcachNSDI}.   A data item is added to the cache after
a client has requested it and failed.  When the cache is full, the LRU caching algorithm
evicts a data item that has not been used for the longest time in order to accommodate the newly requested one.

\subsection{Contributions of this paper}
\noindent (I) 
For a family $\mathcal{H}$ of hash functions under the Simple Uniform Hashing Assumption (SUHA)~\cite{cormen:2001}, we characterize the asymptotic miss ratio of data item requests on a cluster with consistent hashing. 
This asymptotic result, expressed as a conditional probability,  holds almost surely for all random hash functions in $\mathcal{H}$. 
It provides a new analytical framework to study LRU caching with consistent hashing by conditioning on the random hash function, which is an interesting feature 
that most existing asymptotic results do not have.  

\noindent (II)  
We rigorously establish a one-to-one equivalence between a cluster with consistent hashing and a single virtual LRU cache space with a proper size parametrized by the distribution of the random hash function. Conveniently, this equivalence translates the analytic results from the well-studied single server to the complex cluster with consistent hashing. 
Based on this equivalence, we prove the characteristic time approximation for a cluster to characterize the miss probability. 

\noindent(III) Extensive simulations show that our asymptotic results match with the empirical results accurately even for relatively small cache sizes. 

\subsection{Related work}
Consistent hashing has gained much popularity in recent years due to the increasing demand of processing large data sets on a scale-out 
infrastructure. It has been successfully used in a number of real-world applications, e.g., web caches~\cite{Karger:1999,Wang:1999:SWC}, peer-to-peer networks~\cite{stoica2003chord} and distributed storage systems~\cite{lakshman2010cassandra}.
Most theoretical studies on consistent hashing focus on characterizing the randomized partitioning algorithms to balance data allocation~\cite{schickinger2000simplified}, memory sharing~\cite{novakovic2015never} and perfect hashing~\cite{edelkamp2014planning}. These algorithms are usually analyzed under SUHA.
Some works circumvent this assumption using realistic hash functions through simulations~\cite{dietzfelbinger2009applications,pagh2008uniform}.
However, none of these works investigate the miss probabilities of a cluster with consistent hashing.

There is a large body of work on the miss ratio of a single LRU server. 
Different methods have been proposed, e.g.,  approximation by iterative algorithms~\cite{drudi2015approximating}, mean field analysis~\cite{gast2015transient} and the characteristic time approximation~\cite{fagin1977asymptotic,che2002}. 
To obtain insights, asymptotic results for Zipf's popularity distributions have been derived~\cite{Jelenkovic:2004,Jelenkovic99asymptoticapproximation,osogami2010fluid}.   
The characteristic time approximation is also a common approach, which has been shown to be accurate in practical applications~\cite{BergerGSC14,garetto2016unified}.  Its success has been supported by 
the analysis~\cite{Fricker:2012,roberts2013exploring}.
Nevertheless, these results cannot be directly extended to a cluster with consistent hashing. 
 
Although single LRU caches have been studied in depth, 
characterizing the miss behavior of cache networks with general topologies remains difficult. 
Instead, some existing works, e.g.,~\cite{shanmugam2013femtocaching,ioannidis2016adaptive},  focus on offline optimization problems (e.g., content placement) in cache networks.
Some specially structured cache networks (e.g., tree or line networks) have been studied~\cite{berger2014exact,Choungmo2014} using a TTL-based eviction scheme. 
With consistent hashing, the cache network can be viewed as a one-hop network, which is the focus of this paper. 

\section{Model description}\label{s:model}
Consider a cluster $\mathcal{C}$ with $N$ servers $\{S_1,S_2,\cdots,S_N\}$ organized by consistent hashing. 
Assume that the data item requests hosted on $\mathcal{C}$ can access an infinite number of distinct data items of unit size
that are represented by
a sequence $\left(d_i^{\circ},\,i=1,2,3,\cdots\right)$. 
The notation $H(d_i^{\circ})=m$
represents that a data item $d_i^{\circ}$, together with all requests that ask for $d_i^{\circ}$,  are hashed to server $m$ by the hash function $H(\cdot)$. 
Thus, a subsequence of data items, denoted by $( d_{i}^{(m)},i=1,2,\cdots )$,  are hashed to server $S_m$, selected from $\left( d_i^{\circ},\,i=1,2,\cdots\right)$.

To characterize the hash function $H$, we assume the Simple Uniform Hashing Assumption (SUHA)
to facilitate the analysis. In practice, the ideal SUHA property is not feasible, and people resort to a (strongly) universal 
hash family~\cite{Carter:1977} or a $k$-independent hash family~\cite{Wegman79} as approximations. 
Actually, it has been shown that 2-independent hash
functions under mild conditions approximate truly random functions~\cite{Mitzenmacher:2008}.
Specifically, we consider a family of 
hash functions $\mathcal{H}=\{h_w(\cdot), w=1, 2, \cdots \}$. Assume that, with $H$ chosen 
from $\mathcal{H}$ uniformly at random, 
 each data item $d^{\circ}_i, i=1,2, \cdots$ is dispatched also uniformly at random to one of the partitions. 
 Note that although $H(\cdot)$ is random, it becomes one of the deterministic hash function $h_w(\cdot)$ once the random selection is completed.
Using consistent hashing, the involved servers store the data items from mutually exclusive subsets of the partitions, as shown in Fig.~\ref{fig:consisHashing}. 
Since the number of partitions assigned to each of the servers could be different,  
we can equivalently assume that each data item $d^{\circ}_i$ is independently hashed to server $S_m, 1\leq m\leq N$ with probability $\mu_m$ 
by $H(\cdot)$.
In other words,  we have the following assumption. 

\begin{assumption}(SUHA)\label{assump:1}
$H(d^{\circ}_i), i=1, 2, \cdots$ are independent random variables with $\Pr[ H(d^{\circ}_i) = m ]=\mu_m, 1\leq m \leq N$.   
\end{assumption}


Assume that the arrivals of the data item requests
occur at time points $\{\tau_n,-\infty<n<+\infty\}$. Let $J_n$ be the index of the server for the request at time $\tau_n$.  The event $\{J_n=m\}$ represents that the request at time $\tau_n$ is processed on server $S_m$. 
Denote by $R_n$ the requested data item at time $\tau_n$. Thus, the event $\{J_n=m,R_n=d_{i}^{(m)}\}$ means that the request at time $\tau_n$ is to fetch data item $d_{i}^{(m)}$ on server $S_m$.  In order to compute the miss ratio when the system reaches stationarity, we consider the request at time $\tau_0$. It has been shown~\cite{Jelenkovic99asymptoticapproximation}  that the miss ratio 
is equal to the probability that the data item requested by $R_0$ is not in the cache.  
For cluster $\mathcal{C}$, define
\begin{align}
&\Pr\left[R_0=d^{(m)}_{i} {\big |} J_0=m, H \right] = q^{(m)}_{i}, i=1, 2,3,  \cdots\label{eq:conditional}
\\&\Pr\big[R_0=d^{\circ}_{i}\,\big] = q^{\circ}_{i}, i=1, 2, 3,  \cdots\label{eq:unconditional}
\end{align}
Note that $\left(q^{(m)}_{i},i\geq 1\right)$ is a random sequence determined by the random hash function $H$. 
We assume that the data items $\left( d_i^{\circ},i\geq 1\right)$ are sorted such that the sequence $\left( q_{i}^{\circ},\,i\geq 1\right)$ 
is non-increasing with respect to $i$. Since $\left( d_i^{(m)},i\geq 1\right)$ is a random subsequence of $\left( d_i^{\circ},\,i\geq 1\right)$, 
$\left( q_i^{(m)},i\geq 1\right)$ is also non-increasing by this ordering. 
Let $\left(d^{\circ}_{m_i}, i\geq 1\right) \equiv \left( d_i^{(m)},i\geq 1\right)$, which represents the subsequence of $\left( d_i^{\circ}, i \geq 1\right)$ that is hashed to server $S_m$ by $H$. 
Therefore, $\Pr[J_0=m {\big |} H ]=\sum_{i=1}^{\infty}q_{m_i}^{\circ}$. For notational convenience in our proofs, we define $W_m=1/\Pr[J_0=m {\big |} H ]$. 
On server $S_m$, we have $q_i^{(m)}=W_mq_{m_i}^{\circ},\,i=1,2\cdots$. 
We emphasize that $W_m$ is a random variable  determined by the random hash function $H$, 
which normalizes $Q_m=\left( q_{m_i}^{\circ}, i\geq 1\right)$ to be a legitimate distribution. 

The data item popularity is assumed to follow a Zipf's distribution $q_i^{\circ}\sim c_{\circ}/i^{\alpha_{\circ}}$.
This is a typical distribution that has been empirically observed in web pages~\cite{lee99}, content-centric network~\cite{fricker2012impact}, and video systems~\cite{cha2007tube}. 
To simplify the analysis, this paper only considers a Zipf's distribution with 
$\alpha_{\circ}>1$.  For $\alpha_{\circ}<1$, we can conduct a similar analysis based on existing results~\cite{jelenkovic2007lru,guocongINFO,berthet2017approximation}.


LRU is equivalent to the move-to-front (MTF) policy~\cite{Jelenkovic99asymptoticapproximation,fill:1996},  
which sorts the data items in an increasing order of their last access time.
When a data item is requested under MTF, it is moved to the first position of the 
list and all the other data items that were in front of this one increase their positions by one.  
Define $C_n$ to be the position of the data item requested by $R_n$  in the sorted list under MTF on the server that processes the
request $R_n$.
Then, 
the miss probability of the requests on server $S_m$ with a cache size $x_m$ is given by $\Pr[ C_0 > x_m   {\big |} J_0=m, H]$, 
which is conditional on the random hash function $H$ and the event that $R_0$ occurs on server $S_m$, i.e., $J_0=m$. 
Combining the miss ratios of the servers, we obtain the overall miss probability of the cluster $\mathcal{C}$, conditional on $H\in \mathcal{H}$,
\begin{align}\label{eq:misscluster}
\Pr_{miss}^{\,\,\mathcal{C},H}=\sum_{m=1}^{N}\Pr\left[ C_0 > x_m   {\big |}J_0=m, H\right]\Pr\left[ J_0=m {\big |} H\right].
\end{align}
In the analysis, we assume $x_m=b_mx, 1\leq m \leq N$ for $x>0$. 

\section{Main results}\label{se:main}
In this section, we first derive the miss ratio for each of the servers of the LRU cluster with consistent hashing
conditional on the random hash function. 
Then, we show that these individual cache servers 
can be regarded as a single virtual LRU server with a proper cache size. 
This connection also proves the characteristic time approximation for a cluster.

\subsection{Asymptotic miss ratio under random hashing}\label{se:A}


We derive the miss probabilities for the servers of the cluster
 by conditioning on the random hash function $H$. Note that $H$ uniquely determines $W_m, 1\leq m \leq N$.
The gamma function is given by  $\Gamma(\alpha+1)=\int_{0}^{\infty}y^{\alpha}e^{-y}dy$. 
The notation $f(x)\sim g(x)$ means $\lim_{x\rightarrow\infty} f(x)/g(x)=1$.
\begin{theorem}\label{theo:1}
Under the assumptions in Section~\ref{s:model}, 
 we obtain,  for all $1\leq m \leq N$, almost surely for all $H$,  as $x_m\rightarrow\infty$, 
\begin{small}
\begin{align}\label{main:1}
\Pr\left[C_0>x_m\big | J_0=m, H \right]\sim \frac{\left(\mu_m\Gamma(1-1/\alpha_\circ)\right)^{\alpha_\circ}c_\circ W_m}{\alpha_\circ x_m^{\alpha_\circ-1}}, 
\end{align} 
\end{small}
which implies, almost surely for all $H$, 
\begin{align}\label{main:2}
\Pr_{miss}^{\,\,\mathcal{C}, H} \sim \sum_{m=1}^{N}  \frac{ \mu_m^{\alpha_\circ} \Gamma(1-1/\alpha_\circ)^{\alpha_\circ}c_\circ }{\alpha_\circ x_m^{\alpha_\circ-1}}.
\end{align} 
\end{theorem}
\begin{proof}
The proof is presented in Section~\ref{proof_th1}.
\end{proof}
This asymptotic result in (\ref{main:1}) involves random variables $W_m,1\leq m \leq N$ that are determined by $H$. Interestingly, 
the overall asymptotic miss ratio of the whole cluster in~(\ref{main:2}) is independent of $H$ since $W_m=1/\Pr[J_0=m {\big |} H ]$. These asymptotic results hold a.s. for all $H$. See Experiments in Section~\ref{s:simu}. 
If there is only one server in the cluster $\mathcal{C}$, i.e., $N=1$, Theorem~\ref{theo:1} 
reproduces the results in~\cite{Jelenkovic:2004,Jelenkovic99asymptoticapproximation} for a Zipf's distribution, e.g.,  Theorem~3 of~\cite{Jelenkovic99asymptoticapproximation} on an asymptotic miss probability of a single LRU server. However, extending this result
from a single server to a cluster is complicated. We discuss two main issues that cause the difficulty: 1) Theorem~3 of~\cite{Jelenkovic99asymptoticapproximation} assumes a deterministic popularity distribution on a server.  This condition 
is not satisfied in our model due to the random hash function; 
2) the 
proofs of~\cite{Jelenkovic:2004,Jelenkovic99asymptoticapproximation} cannot be used to prove the characteristic time approximation for a cluster. 
Because of these reasons, we use a different approach to derive the miss probability of a LRU cluster with consistent hashing,
which also proves the characteristic time approximation for a cluster. 
Now, suppose that we have a single virtual LRU cache server of size $\bar{x}$ that serves the entire data item requests $\{R_n\}$, which
at the same time are also served on the cluster $\mathcal{C}$. 
Based on 
Theorem~\ref{theo:1}, we establish an equivalence between the cluster $\mathcal{C}$ and the virtual LRU cache. 
Denote by $\Pr\left[ C_0>\bar x \,\, \right ]$ the miss probability of the virtual LRU cache conditional on $H$. 
Recall that the server $S_m$ has a cache capacity $x_m=b_mx$.  
\begin{theorem}\label{theo:3}
Under the assumptions in Section~\ref{s:model}, we obtain, almost surely for all $H$,
\begin{align}\label{hasheq}
\Pr_{miss}^{\,\,\mathcal{C}, H}\sim \Pr\left[ C_0>\bar x  \right ], \; \textrm{as}\; x\to \infty,
\end{align}
 where
\begin{align}\label{ineq:67}
\bar x=x\left(\sum_{m=1}^{N}\mu_m^{\alpha_\circ}b_m^{1-\alpha_\circ}\right)^{-1/(\alpha_\circ-1)}.
\end{align} 
\end{theorem}
\begin{proof}
The proof is presented in Section~\ref{proof_th2}.
\end{proof}
This theorem shows that the miss probability on the cluster $\mathcal{C}$ is asymptotically equal to the miss ratio of a LRU server 
with the cache size given by~(\ref{ineq:67}). Interestingly, as illustrated in Experiment~\ref{ex:2}, this asymptotic equivalence is accurate even when the cache size of each individual server of cluster $\mathcal{C}$ is relatively small.
Using this connection,  existing results and insights that have been established 
for a single server seem to be also true for a LRU cluster with consistent hashing.  This could be useful for resource planning and cluster optimization.  

\subsection{ Characteristic time approximation with consistent hashing}\label{sec:che}
The characteristic time approximation~\cite{che2002} has been widely used in estimating the miss ratio of a LRU server. Based on the connection between a cluster and a single virtual LRU server established by Theorem~\ref{theo:3}, 
we derive the characteristic time approximation for a cluster.  Recall that the server $S_m$ has a cache capacity $x_m=b_mx, 1\leq m \leq N$. 

 Theorem~\ref{theo:1} shows that, although the miss ratio of each server is random, determined by $H$,  the overall asymptotic miss ratio of the cluster
is independent of $H$. This interesting result motivates us to define
the characteristic time approximation for the cluster~$\mathcal{C}$ 
\begin{align}\label{che}
\Pr_{\textrm{CT}}\left[ C_0>\bar x \right]=\sum_{i=1}^{\infty}q_i^{\circ}e^{-q_i^{\circ}t_C},
\end{align}
where $\bar x$ is given by~(\ref{ineq:67}) and $t_C$ is the unique solution of the equation $\sum_{i=1}^{\infty}(1-e^{-q_i^{\circ}t_C})=\bar x$. 
Under the assumptions in Section~\ref{se:A}, we prove that 
\begin{theorem}\label{approx:che}
Under the assumptions of Theorem~\ref{theo:3},  we have, almost surely for all $H$, 
\begin{align}\label{cheche}
\Pr_{\text{CT}}\left[ C_0>\bar x \right]\sim \Pr_{miss}^{\,\,\mathcal{C},H}, \;\; \text{as} \; \bar x\rightarrow\infty.
\end{align}
\end{theorem}
\begin{proof}
The proof is presented in  Section~\ref{proof:theo3}.
\end{proof}

\section{Simulations}\label{s:simu}
In this section, we conduct extensive simulations using C++ to verify the main results in Section~\ref{se:main}.  
Notably, all simulations match with our theoretical results even for relatively small cache sizes. 

\begin{experiment} \label{ex:1}
This experiment verifies Theorem~\ref{theo:1}.  Consider a cluster of $100$ heterogeneous servers $\{S_1,S_2,\cdots, S_{100}\}$ that have distinct cache sizes and different hashing probabilities.
The server $S_m,\,1\leq m\leq 100$ has a  cache capacity $x_m=\left(1+ 0.1 z_m  \right) x$ with $z_m$ selected uniformly at random from $[-0.5,0.5]$. Thus, $x$ is the average cache size across all of the servers.
Recall that $\mu_m$ is the probability that a data item is hashed to server $S_m$.
\begin{figure}[h] 
\vspace{-0.2cm}
\centering
\includegraphics[width=7cm]{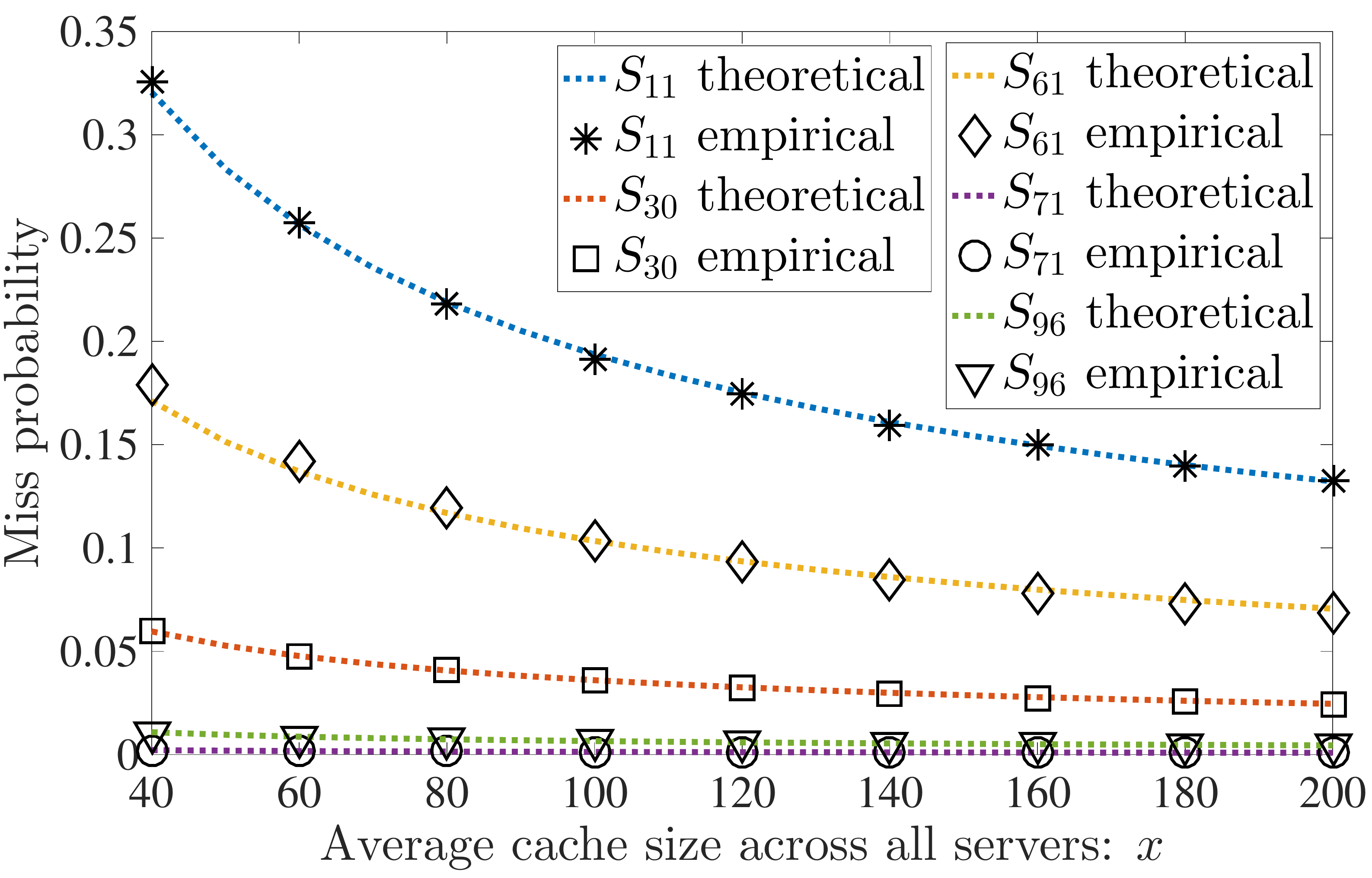}\vspace{-0cm}
\caption{Miss probabilities of five selected servers} \label{fig:ex1}
\vspace{-0.1cm}
\end{figure}
 Let 
$\mu_m=\left( 0.1+0.05 \lfloor (m-1)/20 \rfloor   \right)/20$ for $1\leq m\leq 100$.
 Conditional on $H$, we obtain random variables $W_m, 1\leq m\leq 100$ with $W_m=1/\sum_{i=1}^{\infty}q_{m_i}^{\circ}$. 
Set the total number of data items $M=10^7$, and the popularity distribution $q_i^{\circ}=c_\circ/i^{\alpha_\circ},1\leq i\leq M$ with $\alpha_\circ=1.55, c_\circ=1/(\sum_{i=1}^{M}i^{-\alpha_\circ})= 0.4109$.  For each $x \in \{40,60,80,100,120,140,160,180,200\}$, we first
simulate $10^8$ requests to ensure that the entire system reaches stationary, and then $10^9$ more requests 
to compute the empirical miss probabilities of the cluster $\mathcal{C}$ and the individual servers. To verify~(\ref{main:1}), we need to show
that it holds for all $1\leq m\leq 100$. To visualize the results,  
we only plot the miss probabilities of five servers $\{S_{11},S_{30},S_{61},S_{71},S_{96}\}$ in Fig.~\ref{fig:ex1}.
The empirical results match well with the the theoretical results by~(\ref{main:1}) and~(\ref{main:2}) even when $x$ is small.
\end{experiment}

\begin{experiment}\label{ex:2}
This experiment verifies the equivalence between the cluster $\mathcal{C}$ and a virtual LRU cache
described in Theorem~\ref{theo:3}, by using the same setting as in Experiment~\ref{ex:1}. 
\begin{figure}[h] 
\vspace{-0.2cm}
\centering
\includegraphics[width=8cm]{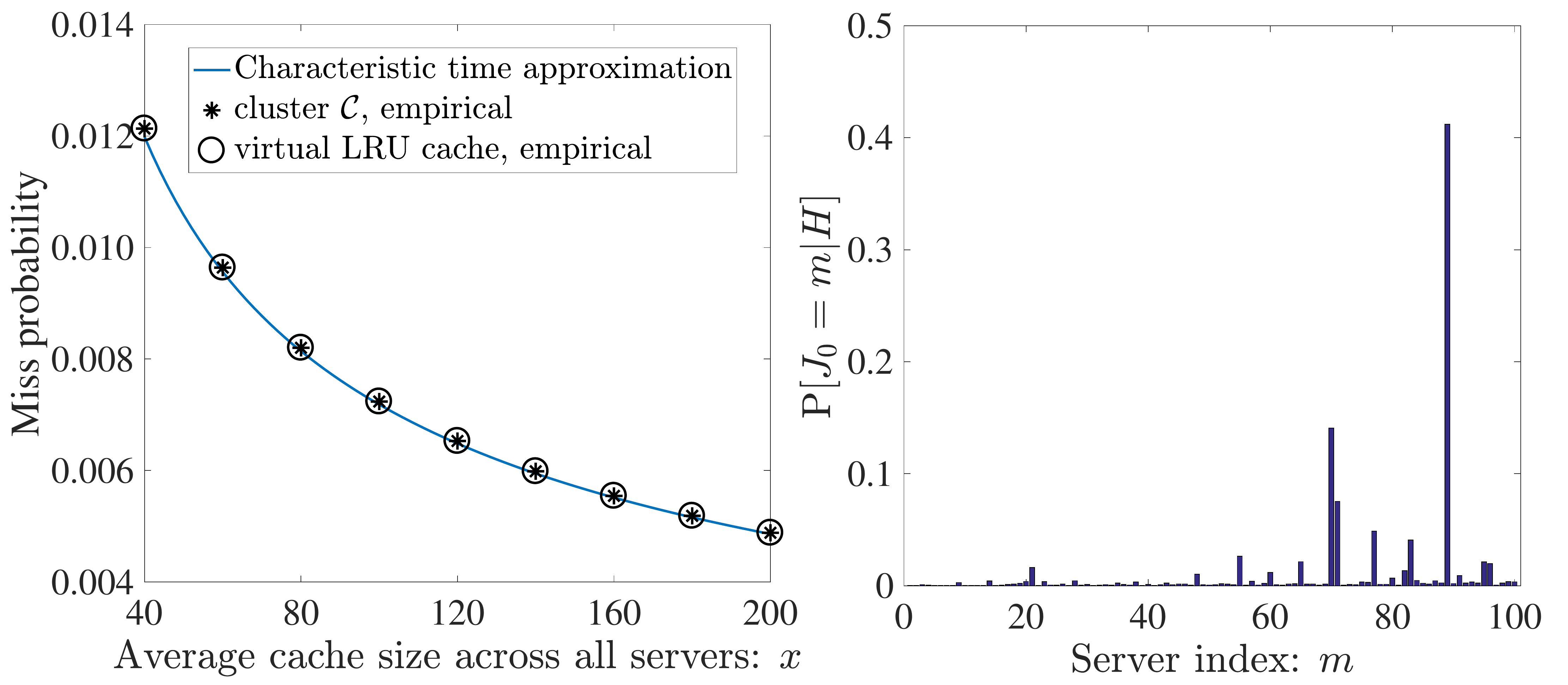}\vspace{-0cm}
\caption{[Left] Miss ratios of the cluster $\mathcal{C}$, a virtual LRU cache and the characteristic time approximation. [Right] The probability to hash a request to server $S_m$ conditional on~$H$} \label{fig:ex2}
\vspace{-0.1cm}
\end{figure}
We demonstrate the accuracy of the characteristic time approximationin~(\ref{che}), which also verifies Theorem~\ref{approx:che}. By computing~(\ref{ineq:67}), we obtain the equivalent size of the virtual LRU caching space $\bar x=91.1784 x$. For the characteristic time approximation of the cluster $\mathcal{C}$, we use a binary search to 
find the solution $t_C$ of the equation $\sum_{i=1}^{\infty}(1-e^{-q_i^{\circ}t_C})=91.1784x$ and then calculate the miss probability by~(\ref{che}).
It can be shown from the left figure in Fig.~\ref{fig:ex2} that the empirical results match well with the theoretical results for the miss probabilities of the virtual LRU caching server and the cluster~$\mathcal{C}$ even for $x=40$. In addition, the characteristic time approximation~(\ref{che}) provides an accurate estimation of the miss ratio of the cluster $\mathcal{C}$. 
\end{experiment}

\begin{experiment}
This experiment moves beyond the assumptions of this paper and considers a realistic setting. 
Thus, we cannot explicitly compute the equivalent virtual cache size by Theorem~\ref{theo:3}. However, we still 
demonstrate an equivalence between the cluster $\mathcal{C}$ and a virtual cache.
We set $\alpha_\circ=0.8$ and use a 2-independent hash 
function~\cite{Mitzenmacher:2008}. For a cluster of $100$ servers $\{S_1,S_2,\cdots, S_{100}\}$ described in Experiment~\ref{ex:1}, 
we hash each server to one of $2000$ partitions using the 2-independent hash function  $h_{a,b}(S_i)= N_i=((a\times i+b) \mod p)\mod 2000$, where
 $a,b$ are chosen 
from $\{1,2,\cdots, p\}$ uniformly at random with a large prime $p=15881$. Using the same hash function, each data item $d_i^\circ$ is hashed to one of these partitions. 
The data items from partition $k$ are stored on the server that has an index $\arg\min_{i}\{N_i: N_i\geq k\}$ if the set $\{i: N_i\geq k\}\neq \emptyset$ and $\arg\min_{i}\{N_i\}$ otherwise.    
\begin{figure}[h] 
\vspace{-0.2cm}
\centering
\includegraphics[width=8cm]{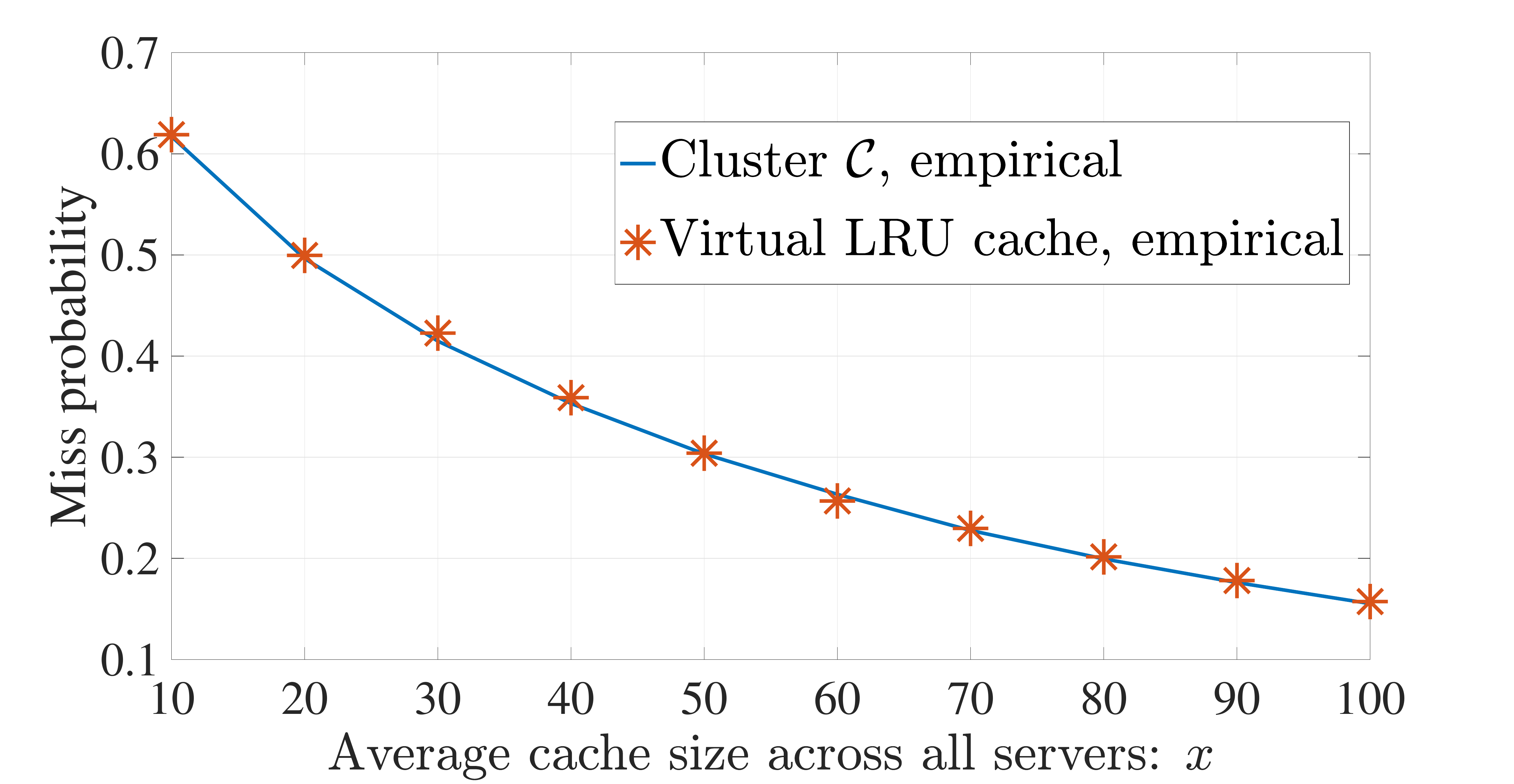}\vspace{0cm}
\caption{ Miss ratios of the cluster $\mathcal{C}$ and a virtual LRU cache} \label{fig:ex5} 
\vspace{-0.1cm}
\end{figure}
We set the total number of data items $M=10^4$, and the popularity distribution $q_i^{\circ}=c_\circ/i^{0.8},1\leq i\leq M$ with $c_\circ=1/(\sum_{i=1}^{M}i^{-0.8})= 0.0369$. 
For the virtual LRU caching space, we find the equivalent size $\bar x=70.0212 x$. It can be shown from Fig.~\ref{fig:ex5} that the empirical miss ratios of the cluster~$\mathcal{C}$ and the virtual LRU cache match very well. 

\end{experiment}

\section{Conclusion}
Driven by the trend to scale out caching systems for processing big data,  LRU caching with consistent hashing has been widely deployed. 
We develop a theoretical framework to investigate the miss ratio of a LRU cluster for a family of hash functions satisfying the Simple Uniform 
Hashing Assumption (SUHA). 
We derive a close-form asymptotic miss probability that holds almost surely for all of the random hash functions from this family. 
This result also establishs a one-to-one equivalence between a LRU cluster and a single virtual LRU server.  It provides 
a convenient abstraction to understand the complex LRU cluster using the insights obtained from a LRU server. 
Based on this connection, we also prove the characteristic time approximation for a cluster with consistent hashing.

\section{Proofs}
This section contains the proofs of our main theorems. 

\subsection{Proof of Theorem~\ref{theo:1}}\label{proof_th1} 

We rely on the following Lemma~\ref{le:1} to prove Theorem~\ref{theo:1}. 
To this end, we use Lemmas~\ref{qiqiplus1},~\ref{phi_m} and~\ref{le:tm} to study the three quantities in Lemma~\ref{le:1}, i.e., $q_i/q_{i+1}$, $\Phi(\cdot)$
and $T(x)$.
Specifically, 
we first note that the variables $\left(I_k^{(m)},k\geq 1\right)$ can be regarded as the indices of data items on server $S_m$ and 
show that $ 1\leq q_{I_k^{(m)}}^{(m)}/q_{I_k^{(m)}+1}^{(m)}\leq 1+\epsilon$ holds with high probability $1-c_1/n^2$ in Lemma~\ref{qiqiplus1}. 
Then, we prove that the functional relationship $\tilde{\Phi}_m(z)$ between $\left( \sum_{i=I_k^{(m)}+1}^{\infty}q_i^{(m)} \right)^{-1}$ and $\Big(q_{I_k^{(m)}}^{(m)}\Big)^{-1} $ satisfies $\tilde{\Phi}_m(z)\sim a z^\beta$ with high probability  $1-c_2/n^2$ in Lemma~\ref{phi_m}. Last, we show $T_{m}(z)\sim c_1z^{\alpha}$ with high probability $1-c_3/n^2$  in Lemma~\ref{le:tm}. Using $\sum_{n=1}^{\infty} (c_1+c_2+c_3)/n^2<\infty$ and Borel-Cantelli lemma,
we prove that Theorem~\ref{theo:1} holds almost surely for all $H\in\mathcal{H}$. 



Lemma~\ref{le:1} is a direct consequence of Theorem~1 in~\cite{jiantanSig}.
Let the data item popularity distribution on a server be $\left( q_i, i\geq 1\right)$. Consider the following functional relationship
\begin{align}\label{eq:relation}
 \Big(\sum_{i=y}^{\infty} q_{i}\Big)^{-1}  \sim \Phi \Big( q_{y}^{-1} \Big), \; y \to \infty.
\end{align}
Define an increasing function $T(x)=\sum_{i=1}^{\infty}\big (1-\big(1-q_{i}\big)^x\big)$ with an inverse $T^{\leftarrow}(x)$. 
We say that $f(x)\lesssim g(x)$ as $x\rightarrow\infty$ if $\limsup_{x\rightarrow\infty} f(x)/g(x)\leq 1$; $f(x)\gtrsim g(x)$ has a complementary definition.
  \begin{lemma}\label{le:1}
If $1\leq \lim_{i\rightarrow\infty}q_i/q_{i+1}<1+\epsilon$, $T(z)\sim c_1z^{\alpha},\alpha>0$ and $\Phi(z)\sim c_2z^\beta,\beta>0$, 
then, as $x\to \infty$,
\begin{align}\label{eq:missP}
\frac{\Gamma(1+\beta)(1+\epsilon)^{-1}}{\Phi(T^{\leftarrow}(x))} \lesssim 
 \Pr[C_0 >x ] \lesssim \frac{\Gamma(1+\beta)(1+\epsilon)}{\Phi(T^{\leftarrow}(x))}.
\end{align}
\end{lemma}
\begin{proof}
Let $\sigma$ be the largest integer such that $R_{-\sigma}=R_0$.  Since the requests $R_{-1},R_{-2},\cdots,R_{-\infty}$ are i.i.d, then we obtain $\Pr\left[\sigma>n \big | R_0=d_i\right]=\left(1-q_i\right)^{n}$, which, in conjunction with $\Pr\left[R_0=d_i\right]=q_i$, indicates that 
\begin{align}\label{qiqi}
\Pr[\sigma>n ]=\sum_{i=1}^{\infty} q_i\left( 1-q_i\right)^n.
\end{align}
In the following step, we show that $\Pr[\sigma>n] \sim \Gamma(\beta+1)/\Phi(n)$.  For simple notions, we define an increasing and continuous function $\bar \Phi(x)=c_2x^\beta,\beta>0$. Noting $\bar \Phi(x)\sim\Phi(x)$ and using~(\ref{eq:relation}),  we obtain,  for $\forall\,\epsilon_1>0$, there exists $n_1\in\mathbb{N}^+$ such that, for $i>n_1$
\begin{align}\label{eq:equivBound2}
\bar \Phi^{\leftarrow}   \left( (1-\epsilon_1) \left( \sum_{j=i}^{\infty} q_j \right)^{-1} \right)  \leq
\left(q_i\right)^{-1}\,\,\,\,\,\, \nonumber
\\\leq  \bar\Phi^{\leftarrow} \left((1+\epsilon_1) \left( \sum_{j=i}^{\infty} q_j \right)^{-1} \right).
\end{align}
where $\bar \Phi^{\leftarrow}(\cdot)$ is the inverse of $\bar \Phi(\cdot)$. Using $\left(1-q_i\right)^n\leq e^{-nq_i}$ in~(\ref{qiqi}), we obtain
\begin{align}\label{shengnan}
&\Pr[\sigma>n ]=\sum_{i=1}^{\infty}q_i\left(1-q_i\right)^n\nonumber
\\&\leq\sum_{i=1}^{n_1}\left(1-q_i\right)^n+\sum_{i=n_1+1}^{\infty}q_ie^{-nq_i}\eqdef P_1+P_2.
\end{align}
For any $0<\epsilon_2<q_{n_1+1}$ and a sufficiently large n, consider a sequence of indices  $i_1>i_2>\cdots>i_m>\cdots>i_{\lfloor \log n\epsilon_2 \rfloor}$ such that $q_{i_k+1}\leq e^{k}/n\leq q_{i_k}$ for $1\leq k\leq \lfloor \log n\epsilon_2 \rfloor$. Then, we have 
\begin{align}
P_2&\leq e^{-n\epsilon_2}+\sum_{i=i_{\lfloor \log n\epsilon_2 \rfloor}}^{i_m}q_ie^{-nq_i}+\sum_{i=i_m+1}^{\infty}q_ie^{-nq_i}\nonumber
\\&\leq e^{-n\epsilon_2}+\sum_{k=m}^{\infty}e^{-e^k}\sum_{j=i_{k+1}+1}^{i_k}q_j+\sum_{i=i_m+1}^{\infty}q_ie^{-nq_i}\nonumber
\\&\eqdef P_{21}+P_{22}+P_{23},\nonumber
\end{align}
where the second inequality follows from the fact that $q_l\geq q_{i_k},\,i_{k+1}+1\leq l\leq i_k $ for $m\leq k\leq \lfloor \log n\epsilon_2 \rfloor$. Letting $Q_i=\sum_{j=i}^{\infty} q_i$, we obtain
\begin{align}
&P_{23}\leq\sum_{j=i_m}^{\infty}(Q_j-Q_{j+1})e^{-n/\bar \Phi^{\leftarrow}((1+\epsilon_1)Q_j^{-1}) }\nonumber
\\&\leq\sum_{j=i_m}^{\infty}\int^{Q_j}_{Q_{j+1}}e^{-n/\bar \Phi^{\leftarrow}((1+\epsilon_1)/x) }dx\nonumber
\\&= \int_{0}^{Q_{i_m}}  e^{-n/\bar \Phi^{\leftarrow}\left((1+\epsilon_1)/x\right) }\,dx\,\,\xleftarrow{\textrm{replace}\,x\,\textrm{by}\, y=\frac{n}{\bar \Phi^{\leftarrow}\left((1+\epsilon_1)/x\right)} }\nonumber
\\&\leq \int_{0}^{\epsilon_1} d\left( \frac{1+\epsilon_1}{\bar\Phi(n/y)} \right) +\int_{\epsilon_1}^{e^m}e^{-y} d\left( \frac{1+\epsilon_1}{\bar\Phi(n/y)} \right),\nonumber 
\end{align}
where the second inequality follows from $e^{-n/\bar \Phi^{\leftarrow}\left((1+\epsilon_1)Q_j^{-1}\right) }$ $\leq$ $ e^{-n/\bar \Phi^{\leftarrow}\left((1+\epsilon_1)/x\right)}$ for $Q_{j+1}\leq x \leq Q_j, j\geq i_m$. Based on Theorem 1.2.1 of~\cite{regularVariation}, we have, if $n\rightarrow\infty$,
\begin{align}
&P_{23}\bar \Phi(n) <(1+\epsilon) \big((1+\epsilon_1)\epsilon_1^{\beta}+(1+\epsilon_1)\int_{\epsilon_1}^{e^m}\beta e^{-y}
y^{\beta-1}dy\big),\nonumber
\\&P_{22}\bar\Phi(n)<(1+\epsilon)(1+\epsilon_1)\sum_{k=m}^{\infty}e^{-e^k}e^{\beta(k+1)}.\nonumber
\end{align}
which, using $\lim_{n\rightarrow\infty}P_1\bar\Phi(n)=0$ and passing $\epsilon_1\rightarrow 0$ and $n,m\rightarrow\infty$, implies that 
\begin{align}\label{ineq:76}
\Pr[\sigma>n ]\bar \Phi(n)<(1+\epsilon)\int_{0}^{\infty}\beta e^{-y}
y^{\beta-1}dy\nonumber
\\=(1+\epsilon)\Gamma(\beta+1).
\end{align}
Since $1\leq \lim_{i\rightarrow\infty}q_i/q_{i+1}<1+\epsilon$, then for any $\epsilon_3>0$, there exists $n_2\in\mathbb{N}^+$ such that for $\forall\,i>n_2$, $q_i>(1-\epsilon_3)q_{i-1}/(1+\epsilon)$, which, in conjunction with~(\ref{eq:equivBound2}) and $q_{i-1}=Q_i-Q_{i-1}$, yields
\begin{align}\label{eigma3}
\Pr[\sigma>n ]&\geq \sum_{i=n_2+1}^{\infty}(1-\epsilon_3)q_{i-1}\left(  1 - q_i \right)^n/(1+\epsilon)\nonumber
\\\geq&\frac{1-\epsilon_3}{1+\epsilon}\int_{0}^{Q_{n_2}} \left(  1- 1/\bar\Phi^{\leftarrow}\left( (1-\epsilon)x^{-1} \right)  \right) dx.
\end{align}
Replacing $x$ by $y=n/\bar\Phi^{\leftarrow}\left( (1-\epsilon)x^{-1} \right)$ in~(\ref{eigma3}) and letting $\bar\Phi^{\leftarrow}\left( (1-\epsilon)Q_{n_2}^{-1} \right)=n/N$, we obtain
\begin{small}
\begin{align}\label{sigmaN}
\Pr[\sigma>n ]\bar\Phi(n)\geq \frac{1-\epsilon_3}{1+\epsilon}\int_{\epsilon}^{N}\left(1-\frac{y}{n}\right)^nd\left( \frac{1-\epsilon}{\bar\Phi(n/y)}        \right).
\end{align}
\end{small}
\hspace{-0.19cm}Then, using a similar approach to~(\ref{ineq:76}) and passing $n\rightarrow\infty$, we obtain from~(\ref{sigmaN}) that $
\Pr[\sigma>n]\bar \Phi(n)>\Gamma(\beta+1)/(1+\epsilon)$,
which, combined with~(\ref{ineq:76}), shows 
\begin{align}\label{ineq:main11}
\frac{(1+\epsilon)^{-1}\Gamma(\beta+1)}{\bar \Phi(n)}\lesssim\Pr[\sigma>n ]\lesssim\frac{(1+\epsilon)\Gamma(\beta+1)}{\bar \Phi(n)}.
\end{align}

Next, we use~(\ref{ineq:main11}) to prove~(\ref{eq:missP}).
Define $D(n)$ as the number of the different data items that have been requested at time $\tau_{-1},\cdots,\tau_{-n}$  and let the inverse of $D(n)$ be $D^{\leftarrow}(x)=\min\{n:D(n)\geq x\}$. It is not hard to show the event $\{C_0>x\}$ is equivalent to $\{  \sigma>D^{\leftarrow}(x)      \}$. 
Thus, we have
\begin{align}\label{eqeqeq}
\Pr[C_0>x ]=\Pr[\sigma>D^{\leftarrow}(x) ].
\end{align}
Define Bernoulli random variables $X_i,i\geq 1$ and let $X_i=1$ imply that the data item $d_i$ has been
requested at time $\tau_{-1},\cdots,\tau_{-n}$  and $0$ otherwise. Let $p_{i,n}\eqdef\Pr[X_i=1]$. Noting $\Pr[X_i=1]=1-(1-q_i)^n$, we obtain $\sum_{i=1}^{\infty}q_{i,n}=T(n)$. For any $\epsilon>0$, using Markov inequality and $\mathbb{E}\left[e^{\epsilon X_i/2}\right]=p_{i,n}\left(  e^{\epsilon/2}  -1\right)+1\leq e^{p_{i,n}(e^{\epsilon/2}-1)}$, we have 
\begin{align}\label{ineq:dmn}
&\Pr\left[  D(n)\geq (1+\epsilon)T(n)    \right]  \leq\mathbb{E}\Big [    e^{\frac{\epsilon}{2}\sum_{i=1}^{\infty}X_i}   \Big ]/e^{\frac{\epsilon}{2}(1+\epsilon) T(n)}\nonumber
\\&\leq \exp\left( \sum_{i=1}^{\infty} p_{i,n}\left(e^{\epsilon/2}-1\right)-\frac{\epsilon}{2}(1+\epsilon)\sum_{i=1}^\infty p_{i,n}\right).
\end{align}
When $\epsilon$ is small enough, we have $e^{\epsilon/2}-1\leq \frac{\epsilon}{2}(\frac{\epsilon}{2}+1)$. Thus, using~(\ref{ineq:dmn}), we obtain
\begin{align}\label{ineq:tmnpr}
\Pr\left[  D(n)\geq (1+\epsilon)T(n)    \right]\leq e^{-\epsilon^2T(n)/4}.
\end{align}
Letting $x=T(n)$ in~(\ref{ineq:tmnpr}) yields
\begin{align}\label{ineq:81}
\Pr\left[  D(T^{\leftarrow}(x))\geq (1+\epsilon)x   \right]\leq e^{-\epsilon^2 x/4}.
\end{align}
Let $\bar x=T^{\leftarrow}\left(x/(1+\epsilon)\right)$. Then, by~(\ref{ineq:81}), we obtain
\begin{align}\label{ineq:pmm}
\Pr\left[  D^{\leftarrow}(x)<\bar x\right ]&\leq \Pr\left[  D\left(T^{\leftarrow}\left(\frac{x}{1+\epsilon}\right)\right)\geq (1+\epsilon)\frac{x}{1+\epsilon} \right]\nonumber
\\ &\leq \exp\left(  -\frac{1}{4}\epsilon^2 x(1+\epsilon)^{-1}   \right).
\end{align}
Letting $\epsilon=x^{-\upsilon},\,\upsilon<1/2$ and using~(\ref{ineq:main11}),~(\ref{eqeqeq}),~(\ref{ineq:pmm}) and a union bound, we obtain, as $x\rightarrow\infty$,
\begin{small}
\begin{align}\label{ineq:finals}
&\Pr\left[  C_0>x \right]\leq\Pr\left[ \sigma>D^{\leftarrow}(x),D^{\leftarrow}(x)\geq \bar x \right]+\Pr\left[ D^{\leftarrow}(x)<\bar x\right]\nonumber
\\&\leq \Pr\left[ \sigma>T^{\leftarrow}\left(\frac{x}{1+x^{-\upsilon}}\right) \right]+ \exp\left(  -\frac{1}{4} x^{1-2\upsilon}(1+x^{-\upsilon})^{-1}   \right)\nonumber
\\ &\lesssim \frac{(1+\epsilon)\Gamma(\beta+1)}{\bar \Phi(T^{\leftarrow}\left(x/(1+x^{-\upsilon)}\right))}+ \exp\left(  -\frac{x^{1-2\upsilon}  }{4(1+x^{-\upsilon})} \right)
\end{align}
\end{small}
\hspace{-0.12cm}Noting $\bar \Phi(T^{\leftarrow}(x))\sim c_1c_2^{-\beta/\alpha}x^{\beta/\alpha}$ and passing $x\rightarrow\infty$, we have, $\bar \Phi(T^{\leftarrow}(x)) \exp\left(  -\frac{1}{4}\epsilon^2 x(1+\epsilon)^{-1}   \right)\rightarrow0$. Thus, passing $x\rightarrow\infty$ and using~(\ref{ineq:finals}), we obatin
\begin{align}\label{ineq:lele}
\Pr\left[  C_0>x \right]\lesssim (1+\epsilon)\Gamma(\beta+1)/\bar \Phi(T^{\leftarrow}\left(x\right)).
\end{align}
Letting $\bar x=T^{\leftarrow}\left(x/(1-\epsilon)\right)$, by a similar approach to~(\ref{ineq:lele}), we obtain, as $x\rightarrow\infty$,
\begin{align}\label{ineq:lele2}
\Pr\left[  C_0>x \right]\gtrsim (1+\epsilon)^{-1}\Gamma(\beta+1)/\bar \Phi(T^{\leftarrow}\left(x\right)).
\end{align}
Combining~(\ref{ineq:lele}) and~(\ref{ineq:lele2}) finishes the proof.
\end{proof}

We introduce some necessary definitions. Define mutually independent Bernoulli random variables $\{X_i^{(m)}\}$. Let
$X_i^{(m)}=1$ indicate that the data item $d_i^{\circ}$ is hashed to server $S_m$ and $X_i^{(m)}=0$ otherwise. 
We have $\Pr[X_i^{(m)}=1]=\mu_m$.  
Define $I_k^{(m)}\eqdef\sum_{i=1}^{k}X_i^{(m)}$, which represents the number of data items hashed to server $S_m$ from 
$(d_i^{\circ},1\leq i\leq k)$. 
Let $Z_i^{(m)}\eqdef q_i^{\circ}X_i^{(m)}$ and $Y_n^{(m)}\eqdef \sum_{i=n}^{\infty}Z_i^{(m)}$.
We quote Bernstein's inequality in Lemma~\ref{bern}, and establish the following Lemma~\ref{le:uB} to estimate $Y_k^{(m)}$, which will be used to estimate $\sum_{i=I_k^{(m)}+1}^{\infty}q_i^{(m)}$ in (\ref{ineq:18}).

%
\begin{lemma}[Theorem $2.8$ in~\cite{chung2006complex}]\label{bern}
For independent random variables $ X_i \leq M,\,1\leq i\leq n$ with $X=\sum_{i=1}^{n}X_i$, we obtain,  $\forall\,\epsilon>0$, 
\begin{align}\label{bern:bound}
\Pr\left[ X-\mathbb{E}[X]>\epsilon \right]\leq \exp\left(-\frac{\epsilon^2}{2\sum_{i=1}^n\mathbb{E}[X_i^2]+\frac{2M\epsilon}{3}}\right).
\end{align}
\end{lemma}

\begin{lemma}\label{le:uB}
There exist $n_0\in\mathbb{N}^+$ and $c_0>0$ such that for $\forall\, n>n_0$, 
\begin{align}
\Pr &\left[  \bigcap_{k\geq n} \left\{ \big |Y_k^{(m)}-\mathbb{E}[Y_k^{(m)}]\big|<1/k^{\alpha_{\circ}-\frac{2}{3}}  \right\} \right]> 1-\frac{c_0}{n^2}.
\end{align}
\end{lemma}
\begin{proof}
Define events $\mathcal{A}_k\eqdef\{ |Y_k^{(m)}-\mathbb{E}[Y_k^{(m)}]|<1/k^{\alpha_{\circ}-\frac{2}{3}}\} ,\,k\geq 1$.  Recalling $q_i^{\circ}$ is non-increasing with respect to~$i$, we have, for all $ i\geq n$ 
\begin{align}\label{cc1}
Z_i^{(m)}\leq q_i^{\circ} \leq q_n^{\circ}\lesssim c_{\circ}/n^{\alpha_{\circ}}.\end{align}
 Noting that $\mathbb{E}[(Z_i^{(m)})^2]=\mu_m(q_i^{\circ})^2\sim\mu_mc_{\circ}^2/i^{2\alpha_{\circ}}$, we have
\begin{align}\label{cc2}
\sum_{i=n}^{\infty}\mathbb{E}[(Z_i^{(m)})^2]\sim\int_{n}^{\infty}\frac{\mu_mc_{\circ}^2}{x^{2\alpha_{\circ}}}dx=\frac{\mu_mc_{\circ}^2}{(2\alpha_{\circ}-1)n^{2\alpha_{\circ}-1}}.
\end{align}
Applying~(\ref{bern:bound}) for the random variable $Y_n^{(m)}=\sum_{i=n}^{\infty}Z_i^{(m)}$ and using~(\ref{cc1}) and~(\ref{cc2}),  we have,
 there exist $n_{i_1}\in\mathbb{N}^+$ and $v>0$ such that for $\forall\,n>n_{i_1}$, 
 \begin{align}\label{pre:re}
 \Pr\left[Y_n^{(m)}-\mathbb{E}[Y_n^{(m)}]>\frac{1}{n^{\alpha_{\circ}-2/3}}\right]<\exp(-n^{1/3}v). 
\end{align}
In the meanwhile,   applying~(\ref{bern:bound}) for $-Y_n^{(m)}=\sum_{i=n}^{\infty}(-Z_i^{(m)})$ and using $-Z_i^{(m)}<0$ and~(\ref{cc2}), we have,  for $ \forall\, n>n_{i_1}$, 
\begin{align}\label{symmetry2}
\Pr&\left[Y_n^{(m)}-\mathbb{E}[Y_n^{(m)}]<-\frac{1}{n^{\alpha_{\circ}-2/3}}\right]
<\exp(-n^{1/3}v),
\end{align}
which, in conjunction with~(\ref{pre:re}), implies, for $n>n_{i_1}$ 
\begin{align}\label{ineq:36}
\Pr\left[\Big |Y_n^{(m)}-\mathbb{E}[Y_n^{(m)}]\Big |\geq\frac{1}{n^{\alpha_{\circ}-2/3}}\right]<2\exp(-n^{\frac{1}{3}}v).
\end{align}
Using (\ref{ineq:36}) and a union bound, we obtain, for $\forall n>n_{i_1}$
\begin{small}
\begin{align}\label{meshi}
\Pr\left[\left\{\cap_{k\geq n}\,\mathcal{A}_k\right\}^c\right]&=\Pr\left[\cup_{k\geq n}\left\{\big |Y_k^{(m)}-\mathbb{E}[Y_k^{(m)}]\big|\geq 1/k^{\alpha_{\circ}-2/3}\right\}\right]\nonumber
\\\leq& \sum_{k=n}^{\infty}\Pr\left[\big |Y_k^{(m)}-\mathbb{E}[Y_k^{(m)}]\big|\geq 1/k^{\alpha_{\circ}-2/3}\right]\nonumber
\\<&\sum_{k=n}^{\infty}2e^{-k^{1/3}v}<\int_{n}^{\infty}\exp(-x^{1/3}v)dx.
\end{align}
\end{small}
\hspace{-0.12cm}There exist $n_{i_2}$ and $c_0>0$ such that $\int_{n}^{\infty}\exp(-x^{1/3}v)dx<c_0/n^2$ holds for $n>n_{i_2}$.
Using~(\ref{meshi}) and letting $n_0\eqdef\max\{n_{i_1},n_{i_2}\}$, we finish the proof.
\end{proof}
We establish the following lemma~\ref{le:7} to estimate $q_{I_k^{(m)}}^{(m)}$, which is used to estimate $\tilde{\Phi}_m(z)$ and the ratio $ q_{I_k^{(m)}}^{(m)}/q_{I_k^{(m)}+1}^{(m)}$. 

\begin{lemma}\label{le:7}
There exist $n_1\in\mathbb{N}^+$ and $c_1>0$ such that for $\forall\, n>n_1$, 
\begin{small}
\begin{align}\label{knm}
\Pr &\left[ \bigcap_{k\geq n} \left\{  W_mq_k^{\circ}\leq q_{I_k^{(m)}}^{(m)}<W_m q^\circ_{k-\lceil k^{1/2} \rceil  +1} 
 \right\} \right]< 1-\frac{c_1}{n^2}.
\end{align}
\end{small}
\end{lemma}
\begin{proof}
Define events $\mathcal{B}_k\eqdef\{ q_{I_k^{(m)}}^{(m)}\geq W_mq^\circ_{k-\lceil k^{1/2} \rceil  +1}\}$. 
Let 
$\mathcal{C}_k\eqdef \left \{  H( d_i^{\circ} ) \neq S_m, k-\lceil k^{1/2} \rceil  +1\leq i\leq k \right\}$ be the event that none
of data items $d_i^\circ, k-\lceil k^{1/2} \rceil  +1\leq i\leq k$ are hashed to $S_m$. 
Next, we will show that $\mathcal{B}_k$ and $\mathcal{C}_k$ are equivalent.

Since both $q_i^{(m)}$ and $q_i^{\circ}$ are non-increasing with respect to~$i$,
the event $\mathcal{B}_k$ implies,  
\begin{align}\label{b1}
d_i^{(m)} \notin   \left \{d_j^\circ,k-\lceil k^{1/2} \rceil  +1\leq j\leq k \right \}, 1\leq i\leq I_k^{(m)}.
\end{align}
Moreover, based on the definition of $I_k^{(m)}$, we have, $\{d_i^{(m)}, i\geq I_k^{(m)}+1  \} \subseteq \{ d_i^\circ,\, i\geq k+1  \}$, which 
implies,
\begin{align}\label{b2}
d_i^{(m)} \notin   \left \{d_j^\circ,k-\lceil k^{1/2} \rceil  +1\leq j\leq k \right \},i\geq I_k^{(m)}+1.
\end{align}
 Combining~(\ref{b1}) and~(\ref{b2}) yields $\mathcal{B}_k \subseteq \mathcal{C}_k$. On the other side, the event $\mathcal{C}_k$ implies $q_i^{(m)}\geq W_mq^\circ_{k-\lceil k^{1/2} \rceil  +1}$ for all $1\leq i\leq I_k^{(m)}$, yieding $\mathcal{C}_k \subseteq\mathcal{B}_k$. Thus, we have $\mathcal{B}_k$ is equivalent to $\mathcal{C}_k$. 
Under Assumption~\ref{assump:1}, we have $\Pr\left[\mathcal{B}_k\right]=\Pr\left[\mathcal{C}_k\right]=(1-\mu_m)^{\lceil k^{1/2}\rceil}$.
Noting the complement $\mathcal{B}_k^c=\big\{   W_mq^\circ_{k-\lceil k^{1/2} \rceil  +1}>q_{I_k^{(m)}}^{(m)} \geq W_mq_k^\circ   \big\}$ and  using a union bound, we obtain,
\begin{align}\label{ineq:bc}
\Pr&\left[\cap_{k\geq n} \mathcal{B}_k^c\right]=1-\Pr\left[\cup_{k \geq n} \mathcal{B}_k\right]\geq 1-\sum_{k\geq n}(1-\mu_m)^{\lceil k^{1/2}\rceil}\nonumber
\\&\geq 1-\sum_{k\geq n}(1-\mu_m)^{k^{1/2}}\geq 1-\int_{n-1}^{\infty} (1-\mu_m)^{x^{1/2}}dx
\end{align}
There exist a large integer $n_1$ and a constant $c_1>0$ such that for all $n\geq n_1$, $\int_{n-1}^{\infty} (1-\mu_m)^{x^{1/2}}dx<c_1/n^2$, which, in conjunction with~(\ref{ineq:bc}), completes the proof.
\end{proof}


We use the following Lemma~\ref{qiqiplus1} to estimate $ q_{I_k^{(m)}}^{(m)}/q_{I_k^{(m)}+1}^{(m)}$. The proof is straightforward by using Lemma~\ref{le:7}.
\begin{lemma}\label{qiqiplus1}
For any $\epsilon_q>0$, there exists  $n_2>n_1$ such that, for all $n>n_2$ 
\begin{align}\label{qiqi1}
\Pr\bigg[\bigcap_{k\geq n}\bigg\{1&\leq q_{I_k^{(m)}}^{(m)}/q_{I_k^{(m)}+1}^{(m)}<1+\epsilon_q \bigg\}\bigg]>1-\frac{2c_1}{n^2}
\end{align}
where the constant $c_1$ is the same as  in Lemma~\ref{le:7}.
\end{lemma}
\begin{proof}
Define an event $D_k\eqdef\Big\{  q_{I_k^{(m)}+1}^{(m)} <  W_m q^{\circ}_{k+\lceil k^{1/2}  \rceil }\Big\}$. 
Noting that $\{d_i^{(m)}, i\geq I_k^{(m)}+1  \} \subseteq \{ d_i^\circ,\, i\geq k+1  \}$, we have  $q_{I_k^{(m)}+1}^{(m)}<W_mq_k^{\circ}$, and hence the complement $D_k^c=\Big\{  W_m q^{\circ}_{k+\lceil k^{1/2}  \rceil }<q_{I_k^{(m)}+1}^{(m)}  < W_mq_k^{\circ} \Big\}$. 
Using a similar approach to Lemma~\ref{le:7}, we obtain $\Pr\left[  D_k  \right]=(1-\mu_m)^{\lceil k^{1/2}\rceil}$ and  for all $n>n_1$
\begin{align}
\Pr\left[ \cap_{k\geq n}\mathcal{D}_k^c\right]&=1-\Pr\left[ \cup_{k\geq n}\mathcal{D}_k\right]\geq 1-\sum_{k\geq n}\Pr\left[\mathcal{D}_k\right]\geq 1-\frac{c_1}{n^2}.\nonumber
\end{align}
which,  togehter with~(\ref{ineq:bc}), we obtain
 \begin{align}
 \Pr\bigg [&\bigcap_{k\geq n}  \bigg\{  W_mq_k^{\circ}>q_{I_k^{(m)}+1}^{(m)}>W_m q^{\circ}_{k+\lceil k^{1/2} \rceil },\,\,\,\,\,\,\,\,\,\nonumber
 \\ &W_mq^\circ_{k-\lceil k^{1/2} \rceil  +1}>q_{I_k^{(m)}}^{(m)}>W_mq_k^\circ\bigg\} \bigg ]  >1-\frac{2c_1}{n^2}, \nonumber
 \end{align}
which implies
\begin{align}
\Pr\bigg[\bigcap_{k\geq n}\bigg\{1&\leq \frac{q_{I_k^{(m)}}^{(m)}}{q_{I_k^{(m)}+1}^{(m)}}<\frac{q^\circ_{k-\lceil k^{1/2} \rceil  +1}}{q^{\circ}_{k+\lceil k^{1/2}\rceil}}\bigg\} \bigg]>1-\frac{2c_1}{n^2}. \nonumber
\end{align}
Using $q^\circ_{k-\lceil k^{1/2} \rceil  +1}/q^{\circ}_{k+\lceil k^{1/2}\rceil} \to 1$, we finish the proof.
\end{proof}


Next, we establish the following lemma to estimate the functional relationship between $\Big( \sum_{i=I_k^{(m)}+1}^{\infty}q_i^{(m)} \Big)^{-1}$ and $\left(q_{I_k^{(m)}}^{(m)}\right)^{-1} $
based on Lemma~\ref{le:uB}  and Lemma~\ref{le:7}. 
\begin{lemma}\label{phi_m}
 For
$\epsilon_p>0$, there exists $n_3>\max\{n_0,n_1\}$ and $c_2>0$ such that, for all $n>n_3$, 
\begin{align}\label{ineq:phiphi}
&\Pr\Bigg[\bigcap_{k\geq n} \Bigg\{(1-\epsilon_p)\Bigg( \sum_{i=I_k^{(m)}+1}^{\infty}q_i^{(m)} \Bigg)^{-1}<\tilde \Phi_m\left(\left(q_{I_k^{(m)}}^{(m)}\right)^{-1} \right)\nonumber
\\&\,\,\,\,\,\,\,\,\,\,\,\,\,\,\,<(1+\epsilon_p)\Bigg( \sum_{i=I_k^{(m)}+1}^{\infty}q_i^{(m)} \Bigg)^{-1}\Bigg\}\Bigg ]>1-\frac{c_2}{n^2}, 
\end{align}
where
\begin{align}\label{ineq:tilphi}
\tilde \Phi_m(x)=(W_mc_\circ)^{-1/\alpha_\circ} \mu_m^{-1} x^{1-1/\alpha_\circ} (\alpha_\circ-1),\,\alpha_\circ>1.
\end{align}
\end{lemma}
\begin{proof}
Since $q_k^{\circ}\sim c_\circ/k^{\alpha_\circ}$, there exist constants $\epsilon_k>0,k\geq 1$ satisfying $\lim_{k\rightarrow \infty}\epsilon_k\rightarrow\infty$ such that 
\begin{align}\label{qkqkqk}
q_k^{\circ}>\frac{(1-\epsilon_k)c_\circ}{k^{\alpha_\circ}},\,q^\circ_{k-\lceil k^{1/2} \rceil  +1}<\frac{(1+\epsilon_k)c_\circ}{(k-\lceil k^{1/2} \rceil  +1)^{\alpha_\circ}}.
\end{align} 
Combining~(\ref{knm}) and~(\ref{qkqkqk}), we obtain, for $n>n_1$ 
\begin{align}\label{ineq:26}
\Pr\Bigg[\bigcap_{k\geq n} \Bigg\{\frac{\left(  k-\lceil k^{1/2} \rceil  +1  \right)^{\alpha_\circ}}{(1+\epsilon_k)W_mc_\circ}<\left(q_{I_k^{(m)}}^{(m)}\right)^{-1}\,\,\,\,\,\,\,\,\,\nonumber 
\\<\frac{k^{\alpha_\circ}}{(1-\epsilon_k)W_mc_\circ}\Bigg\} \Bigg]\geq 1-\frac{c_1}{n^2}.
\end{align}
Noting that $\{d_i^{(m)}, i\leq I_k^{(m)}  \} \subseteq \{ d_i^\circ,\, i\leq k  \}$ and $\{d_i^{(m)}, i\geq I_k^{(m)}+1  \} \subseteq \{ d_i^\circ,\, i\geq k+1  \}$, we have
\begin{align}\label{ineq:18}
W_mY_{n+1}^{(m)}=\sum_{i=n+1}^\infty W_mq_i^{\circ}X_i^{(m)}=\sum_{i=I_n^{(m)}+1}^{\infty}q_i^{(m)}. 
\end{align}
Combining~(\ref{ineq:18}) and Lemma~\ref{le:uB}, we obtain, for $n>n_0$, 
\begin{small}
\begin{align}\label{ineq:27}
\Pr\Bigg[&\bigcap_{k\geq n}\Bigg\{\left(W_m\left( \mathbb{E}\left[Y_{k+1}^{(m)}\right]+1/k^{\alpha_\circ-\frac{2}{3}}  \right)\right)^{-1}<\Bigg( \sum_{i=I_k^{(m)}+1}^{\infty}q_i^{(m)}     \Bigg)^{-1}\nonumber
\\&<\left(W_m\left( \mathbb{E}\left[Y_{k+1}^{(m)}\right]-1/k^{\alpha_\circ-\frac{2}{3}}  \right)\right)^{-1}\Bigg\}\Bigg ]\geq 1-\frac{c_0}{n^2}.
\end{align}
\end{small}
\hspace{-0.12cm}Since $\mathbb{E}\big[Y_{k+1}^{(m)} \big]=\mu_m \sum_{i=k+1}^{\infty}q_i^\circ\sim \mu_mc_\circ(\alpha_\circ-1)^{-1}(k+1)^{1-\alpha_\circ}$, 
there exist constants $\epsilon_{1,k}>0$ with $\lim_{k\rightarrow\infty}\epsilon_{1,k}=0$ such that 
\begin{small}
\begin{align}\label{ineq:28}
\frac{(1-\epsilon_{1,k})\mu_mc_\circ}{(\alpha_\circ-1)(k+1)^{\alpha_\circ-1} }<\mathbb{E}[Y_{k+1}^{(m)}]<\frac{(1+\epsilon_{1,k})\mu_mc_\circ}{(\alpha_\circ-1)(k+1)^{\alpha_\circ-1} }.
\end{align}
\end{small} 
\hspace{-0.12cm}Combining~(\ref{ineq:27}) and~(\ref{ineq:28}), we have, there exist constants $\tau_{1,k},$ $\tau_{2,k}$ satisfying $\lim_{k\rightarrow\infty}\tau_{1,k},\tau_{2,k}\rightarrow 1$ such that, for $n> n_0$, 
\begin{align}\label{ineq:29}
&\Pr\Bigg[\bigcap_{k\geq n} \Bigg\{\frac{(\alpha_\circ-1)k^{\alpha_\circ-1}}{\mu_mW_mc_\circ}\tau_{1,k}<\Bigg( \sum_{i=I_k^{(m)}+1}^{\infty}q_i^{(m)}     \Bigg)^{-1}\nonumber
\\&\,\,\,\,\,\,\,\,\,\,\,\,\,\,\,\,<\frac{(\alpha_\circ-1)k^{\alpha_\circ-1}}{\mu_mW_mc_\circ}\tau_{2,k}\Bigg\}\Bigg]>1-\frac{c_0}{n^2},
\end{align}
Combining~(\ref{ineq:tilphi}) and~(\ref{ineq:26}) implies, for $n>\max\{n_0,n_1\}$,
\begin{align}\label{ineq:32}
\Pr \Bigg[\bigcap_{k\geq n} \Bigg\{\frac{(\alpha_\circ-1)k^{\alpha_\circ-1}}{\mu_mW_mc_\circ}\delta_{1,k}
<\tilde \Phi_m\left(\left(q_{I_k^{(m)}}^{(m)}\right)^{-1} \right)\nonumber
\\<\frac{(\alpha_\circ-1)k^{\alpha_\circ-1}}{\mu_mW_mc_\circ}\delta_{2,k}\Bigg\}\Bigg]>1-\frac{c_1}{n^2}
\end{align}
where constants $\delta_{1,k},\,\delta_{2,k}\rightarrow 1$ as $k\rightarrow\infty$.
Using a union bound to~(\ref{ineq:29}) and~(\ref{ineq:32}), we obtain, for $n>\max\{n_0,n_1\}$, 
\begin{align}\label{tilphi}
&\Pr\Bigg[\bigcap_{k\geq n} \Bigg\{\frac{\delta_{1,k}}{\tau_{2,k}}\Bigg( \sum_{i=I_k^{(m)}+1}^{\infty}q_i^{(m)} \Bigg)^{-1}<\tilde \Phi_m\left(\left(q_{I_k^{(m)}}^{(m)}\right)^{-1} \right)\nonumber
\end{align}
\begin{align}
&\,\,\,\,\,\,\,\,\,\,\,\,\,\,\,<\frac{\delta_{2,k}}{\tau_{1,k}}\Bigg( \sum_{i=I_k^{(m)}+1}^{\infty}q_i^{(m)} \Bigg)^{-1}\Bigg\}\Bigg ]>1-\frac{c_0+c_1}{n^2}
\end{align}
Since $\delta_{1,k}/\tau_{2,k}$, $\delta_{2,k}/\tau_{1,k}\rightarrow1$ as $k\rightarrow\infty$, for any $\epsilon_p>0$,  
there exists $n_3>\max\{n_0,n_1\}$ such that, for any $n>n_3$, $1-\epsilon_p<\delta_{1,k}/\tau_{2,k}, \delta_{2,k}/\tau_{1,k}<1+\epsilon_p$, which, together with~(\ref{tilphi}), completes the proof.
\end{proof}

To use Lemma~\ref{le:1}, we define $T_m(x)=\sum_{i=1}^{\infty}\big (1-\big(1-q^{(m)}_{i}\big)^x\big)$, which is equivalent to 
\begin{align}\label{epsi3}
T_m{(x)}=\sum_{i=1}^{\infty}\left (1-(1-W_mq_{i}^{\circ})^x\right )X_i^{(m)}.
\end{align}
We now derive an approximation of $T_{m}(\cdot)$ using  Lemma~\ref{le:tm}. 
On the proof, we first rewrite  random variable $1/W_m=W_{m_1}+W_{m_2}$, where $W_{m_1}=\sum_{i=1}^{n}q_i^\circ X_i^{(m)}$ and $W_{m_2}=\sum_{i=n+1}^{\infty}q_i^\circ X_i^{(m)}$. For $W_{m_1}$, using Lemma~\ref{le:7}, we show that for $n$ large enough, $W_{m_2}\approx c(n)^{1-\alpha_\circ}$ with high probability. For $W_{m_1}$, note that, conditional on $\{H(d_i^{\circ}),1\leq i\leq n\}$,  $W_{m_1}$ is deterministic. Combining these two results, we have,  for $n$ large enough, conditional on $\{H(d_i^{\circ}),1\leq i\leq n\}$, $W_m$ is deterministic. Based on this fact and conditional on $\{H(d_i^{\circ}),1\leq i\leq n\}$,  we apply the Bernstein's inequality~(\ref{bern:bound}) for~(\ref{epsi3}) and obtain the estimation~(\ref{folo}). By noting that the bound in~(\ref{folo}) is independent of the hashing function $H$, unconditional on $\{H(d_i^{\circ}),1\leq i\leq n\}$, we finish the proof. 


\begin{lemma}\label{le:tm}
For any $\epsilon_2>0$, there exists $x_0>0$,  such that for all $x\geq x_0$, 
\begin{align}\label{ineq:tinverse}
&\Pr\bigg[\bigg \{(1-\epsilon_2)\mu_m\Gamma(1-1/\alpha_\circ)\left(c_\circ W_{m}x\right)^{\frac{1}{\alpha_\circ}}<T_m(x)\nonumber
\\<&(1+\epsilon_2)\mu_m\Gamma(1-1/\alpha_\circ)\left(c_\circ W_{m}x\right)^{\frac{1}{\alpha_\circ}}\bigg\}\bigg]>1-\frac{c_3}{x^2},
\end{align}
where $c_3$ is  a positive constant. 
\end{lemma}
\begin{proof}
For $n_x=\lfloor x^{\sigma}\rfloor$ with $ \sigma<\alpha^{-1}_0$, we define
 $W_{m_1}=\sum_{i=1}^{n_x}q_i^\circ X_i^{(m)}$ and $W_{m_2}=\sum_{i=n_x+1}^{\infty}q_i^\circ X_i^{(m)}$. 
Recalling the definition of $Y_n^{(m)}$, we have $W_{m_2}=Y_{n_x+1}^{(m)}$, which, together with
(\ref{ineq:36}) and~(\ref{ineq:28}), we obtain, for $\epsilon>0$, there exists a large $x_1>0$ such that for all $x>x_1$  
\begin{align}\label{epsis1}
\Pr\bigg[(1-\epsilon)K(n_x)<W_{m_2}<(1+\epsilon)K(n_x)\bigg]\nonumber
\\>1-2\exp(-n_x^{1/3}v),
\end{align}
where $K(n_x)\eqdef\mu_mc_\circ(\alpha_\circ-1)^{-1}(n_x+1)^{1-\alpha_\circ}$ and $\nu$ is the same constant as in~(\ref{ineq:36}). Noting $W_m=1/(W_{m_1}+W_{m_2})$ and using~(\ref{epsi3}) and~(\ref{epsis1}), we have
\begin{align}\label{epsisisi}
\Pr&\bigg[I_1+I_2<T_m(x)<I_1+I_3 \bigg]>1-2\exp(-n_x^{\frac{1}{3}}v).
\end{align} 
where $I_1=\sum_{i=1}^{n_x}\left (1-\left(1-W_{m}q_{i}^{\circ}\right)^x\right )X_i^{(m)}$, 
\begin{small}
\begin{align}\label{I1I2}
I_3=\sum_{i=n_x+1}^{\infty}\left (1-\left(1-\frac{q_{i}^{\circ}}{W_{m_1}+(1-\epsilon)K(n_x)}\right)^x\right )X_i^{(m)}. 
\end{align}
\end{small}
\hspace{-0.15cm}and $I_2$  is defined by replacing $(1-\epsilon)$ in $I_3$ with $(1+\epsilon)$.

\hspace{-0.121cm}Let  $M_{n_x}\eqdef\left \{H(d_i^\circ)=S_{m_i},1\leq i\leq n_x\right\}$ be an event that the first $n_x$ data items are hashed to servers $S_{m_1},\dots, S_{m_{n_x}}$.
Let $p_i=q_i^{\circ}/(W_{m_1}+(1-\epsilon)K(n_x))$. 
From~(\ref{epsi3}), we have $0<W_mq_i^\circ<1$ if $X_i^{(m)}=1$ and $(1-(1-W_mq_{i}^{\circ})^x )X_i^{(m)}=0$ otherwise. Thus, without changing the expression of $T_m(x)$, we assume $0<W_mq_i^\circ <1$ for all $i\geq 1$.
Note that $p_i\rightarrow W_mq_i^{\circ}$ as $n_x\rightarrow\infty$.
 Then, for $n_x$ large enough and conditional on  $M_{n_x}$, we have
$p_i\sim c_1/i^{\alpha_{\circ}}$ and $0<p_i<1$, where the constant $c_1=c_{\circ}/(W_{m_1}+(1-\epsilon)K(n_x))$.
Similar to the derivation of Lemma~\ref{le:uB},  
applying Lemma~\ref{bern} to $I_3x^{-1/\alpha_\circ}=\sum_{i=n_x+1}^{\infty}\left (1-\left(1-p_i\right)^x\right )x^{-1/\alpha_{\circ}}X_i^{(m)}$ and using Lemma 1 in~\cite{jiantanSig}, we obtain, for $\forall\,\epsilon_1>0$, there exists a large $x_2$ such that for all $x>x_2$,
\begin{align}\label{epsi1}
&\Pr\bigg[(1-\epsilon_1)\mu_m\Gamma(1-1/\alpha_\circ)\left(\frac{c_\circ x}{W_{1-\epsilon}}\right)^{\frac{1}{\alpha_\circ}}<I_3<(1+\epsilon_1)\mu_m\nonumber
\\&\Gamma(1-1/\alpha_\circ)\left(\frac{c_\circ x}{W_{1-\epsilon}}\right)^{\frac{1}{\alpha_\circ}}\,\bigg | \,M_{n_x}\bigg]>1-\frac{c_4}{x^2},
\end{align}
where $W_{1-\epsilon}=W_{m_1}+(1-\epsilon)K(n_x)$ and $c_4$ is a positive constant.
A similar result holds for $I_2$ by replacing $(1-\epsilon)$ and $c_4$ in (\ref{epsi1}) with $(1+\epsilon)$ and $c_5$, respectively. 
Recalling the definition~(\ref{epsi3}) and conditional on the event $M_{n_x}$, we have $I_1\leq n_x$, which, in conjunction with $n_x<x^{\sigma}$ and $\sigma<1/\alpha_{\circ}$,  implies that $\lim_{x\rightarrow\infty}{I_1 /x^{\frac{1}{\alpha_{\circ}}}}=0$. Then, conditional on $M_{n_x}$, we have 
 for $\forall\,\epsilon_2>\epsilon_1$, there exists a sufficiently large $x_3$  such that for all $x>x_3$,
\begin{align}
0<I_1<(\epsilon_2-\epsilon_1)\mu_m\Gamma(1-1/\alpha_\circ)\left(c_\circ x/W_{1-\epsilon}\right)^{\frac{1}{\alpha_\circ}},\nonumber
\end{align}
which, using~(\ref{epsisisi}) and~(\ref{epsi1}) and a union bound,  implies that for all $x>x_4\eqdef\max\{x_1,x_2,x_3\}$,
\begin{align}\label{folo}
&\Pr\bigg[(1-\epsilon_2)\mu_m\Gamma(1-1/\alpha_\circ)\left(\frac{c_\circ x}{W_{1+\epsilon}}\right)^{\frac{1}{\alpha_\circ}}<T_m(x)<(1+\epsilon_2)\nonumber
\\&\mu_m\Gamma(1-1/\alpha_\circ)\left(\frac{c_\circ x}{W_{1-\epsilon}}\right)^{\frac{1}{\alpha_\circ}}\,\bigg | \,M_{n_x}\bigg]>1-\frac{c_3}{x^2},
\end{align}
where $W_{1+\epsilon}=W_{m_1}+(1+\epsilon)K(n_x)$ and $c_3$ is a positive constant related to $c_4$, $c_5$ and $\nu$.
Note that $1-c_3/n^2$ on the right side of~(\ref{folo}) is independent of $M_{n_x}$ and $W_{1+\epsilon},\,W_{1-\epsilon}\rightarrow W_m$ as $x\rightarrow\infty$. Thus,  for $x$ large enough, unconditional on $M_{n_x}$ and  passing $\epsilon_1,\epsilon\rightarrow 0$, we obtain~(\ref{ineq:tinverse}).
\end{proof}

%

Now, we use Lemmas \ref{le:1}, \ref{qiqiplus1}, \ref{phi_m} and \ref{le:tm} to prove Theorem~\ref{theo:1} by applying the Borel-Cantelli lemma.
\begin{proof}[Proof of Theorem~\ref{theo:1}]
 For the two bounds proved in Lemma~\ref{qiqiplus1} and Lemma~\ref{phi_m}, it is easy to verify that 
$ \sum_{n=1}^{\infty} \left(2c_1+c_2\right)/n^2  < \infty$,
 which, by the Borel-Cantelli lemma, implies that,  almost surely for each $H \in \mathcal{H}$, there exists a finite $n_H$,
  such that for all $n>n_H$, we have  $1\leq q_{I_n^{(m)}}^{(m)}/q_{I_n^{(m)}+1}^{(m)} < 1+\epsilon_q$ and $(1-\epsilon_p)\Big( \sum_{i=I_k^{(m)}+1}^{\infty}q_i^{(m)} \Big)^{-1}<\tilde \Phi_m\Big(\Big(q_{I_k^{(m)}}^{(m)}\Big)^{-1} \Big)
<(1+\epsilon_p)\Big( \sum_{i=I_k^{(m)}+1}^{\infty}q_i^{(m)} \Big)^{-1}$, 
  where $\tilde \Phi_m$ is defined in Lemma~\ref{phi_m}.
 
 Now, note that $x$  in Lemma~\ref{le:tm} represents the cache size, which also takes integer values. For the bound in Lemma~\ref{le:tm}, 
 we have
   $\sum_{x=1}^{\infty} c_3/x^2 < \infty$, 
 which, by Borel-Cantelli lemma, implies that, almost surely for each $H \in \mathcal{H}$, there exists a finite $x_H$,
  such that for all $x>x_H$, we have
$(1-\epsilon_2)\mu_m\Gamma(1-1/\alpha_\circ)\left(c_\circ W_{m}x\right)^{\frac{1}{\alpha_\circ}}<T_m(x) <(1+\epsilon_2)\mu_m\Gamma(1-1/\alpha_\circ)\left(c_\circ W_{m}x\right)^{\frac{1}{\alpha_\circ}}$. 

These two facts,  using Lemma~\ref{le:1} and passing $\epsilon_q,\epsilon_p, \epsilon_2 \to 0$,  implies that,   almost surely for all $H$, 
\begin{align}
\Pr\left[C_0>x_m\big | J_0=m, H \right]\sim &\frac{\left(\mu_m\Gamma(1-1/\alpha_\circ)\right)^{\alpha_\circ}c_\circ W_m}{\alpha_\circ x_m^{\alpha_\circ-1}}, \nonumber
\end{align}
which finishes the proof of Theorem~\ref{theo:1}.
\end{proof}

\subsection{Proof of Theorem~\ref{theo:3}}\label{proof_th2}
\begin{proof}
By Theorem~\ref{theo:1}, the miss ratio of  cluster $\mathcal{C}$ satisfies
\begin{align}\label{shengnan3}
\Pr_{miss}^{\,\,\mathcal{C},H}&=\sum_{m=1}^{N}\Pr\left[ C_0 > x_m   {\big |}J_0=m,H\right]\Pr\left[ J_0=m\big | H\right]\nonumber
\\&   \sim  \sum_{i=1}^{N} \frac{\left(\mu_m\Gamma(1-1/\alpha_\circ)\right)^{\alpha_\circ}c_\circ}{\alpha_\circ b_m^{\alpha_\circ-1}x^{\alpha_\circ-1}},  
\end{align}
which holds almost surely for all $H$.
Noting that $q_i^{\circ}\sim c_\circ/i^{\alpha_\circ}$ and using Theorem~3 of~\cite{Jelenkovic99asymptoticapproximation}, we can obtain
$\Pr [ C_0>\bar x   ]\sim c_\circ (\Gamma(1-\alpha_\circ^{-1}) )^{\alpha_\circ}/(\alpha_\circ \bar x^{\alpha_\circ-1})$,
which, in conjunction with~(\ref{ineq:67}),  yields~(\ref{hasheq}).
\end{proof}
\subsection{Proof of Theorem~\ref{approx:che}}\label{proof:theo3}
To prove Lemma~\ref{le:che}, we need the following lemma.
\begin{lemma}\label{le:1.5}
If $p_i\sim c/i^\alpha,\,i\geq 1$, then we have, for  any fixed $n_0\in\mathbb{N}^+$, 
\begin{align}\label{TMX}
\sum_{i=n_0}^{\infty}\left( 1-\left(  1-p_i  \right)^x   \right)\sim (cx)^{1/\alpha}\Gamma(1-1/\alpha).
\end{align} 
\end{lemma}
\begin{proof}
By the equality $\big (1-p_i\big )^x\leq e^{-p_ix}$, we have 
\begin{align}\label{ineq:ppi}
\sum_{i=1}^{\infty}\left( 1-\left(  1-p_i  \right)^x   \right)\geq\sum_{i=0}^{\infty}\left(1-e^{-q_i^{(m)}x}\right).
\end{align}
For $\forall\,\delta>0$, there exists $m_\delta$ such that $1-x\geq\exp(-(1+\delta)x)$ for $0\leq x\leq m_\delta$. Since $\lim_{i\rightarrow\infty}p_i\rightarrow 0$, then we can choose large $i_0$ such that for any $i>i_0$, $p_i<m_\delta$. Then, we have 
\begin{align}\label{ineq:i0}
\sum_{i=1}^{\infty}\left( 1-\left(  1-p_i  \right)^x   \right)\leq\sum_{i=1}^{i_0}1+\sum_{i=i_0+1}^{\infty}\left(1-\left(1-p_i\right)^x\right)\nonumber
\\\leq i_0+\sum_{i=1}^{\infty}\left(1-e^{-(1+\delta)p_ix}\right).
\end{align}
Since we have 
\begin{align}
\sum_{i=1}^{\infty}\left(1-e^{-p_ix}\right)\sim \int_{1}^{\infty} \left(1-e^{-cx/z^{\alpha}}\right)dz\sim (cx)^{\frac{1}{\alpha}}\Gamma(1-\frac{1}{\alpha}),\nonumber
\end{align}
then, by passing $x\rightarrow\infty$, we obtain
\begin{align}\label{ineq:dede}
\sum_{i=1}^{\infty}\left(1-e^{-(1+\delta)p_ix}\right)/\sum_{i=1}^{\infty}\left(1-e^{-p_ix}\right)\sim (1+\delta)^{1/\alpha}
\end{align}
Combining~(\ref{ineq:ppi}),~(\ref{ineq:i0}) and~(\ref{ineq:dede}) and passing $\delta\rightarrow 0$, we have 
\begin{align}\label{jkjkjk}
\sum_{i=1}^{\infty}\left( 1-\left(  1-p_i  \right)^x   \right)\sim\sum_{i=0}^{\infty}\left(1-e^{-q_i^{(m)}x}\right)\sim (cx)^{\frac{1}{\alpha}}\Gamma(1-\frac{1}{\alpha})
\end{align}
which, in conjunction with that $\sum_{i=1}^{n_0-1}\left( 1-\left(  1-p_i  \right)^x   \right)<n_0-1$, finishes the proof.
\end{proof}

\begin{lemma}\label{le:che}
If $q_i^{\circ}\sim c_\circ/i^{\alpha_\circ}$, then we have, as $\bar x\rightarrow\infty$, 
\begin{align}
\Pr_{\textrm{CT}}\left[ C_0>\bar x \right]\sim \Pr\left[ C_0>\bar x  \right].\nonumber
\end{align}
\end{lemma}
\begin{proof}
Recall 
\begin{align}\label{jikaiyi}
\Pr_{\textrm{CT}}\left[ C_0>\bar x \right]=\sum_{i=1}^{\infty}q_i^{\circ}e^{-q_i^{\circ}t_C},
\end{align}
where $t_C$ is the unique solution to the equation $\sum_{i=1}^{\infty}(1-e^{-q_i^{\circ}t_C})=\bar x$. 
Noting $\lim_{i\rightarrow\infty}q_i^\circ/q_{i+1}^\circ=1$, we can let $\epsilon=0$ in Lemma~\ref{le:1}.
 Then, using~(\ref{jkjkjk}) and Lemma~\ref{le:1}, we obtain $t_C\sim T^{\leftarrow}(\bar x)$, which, in conjunction with~(\ref{qiqi}),~(\ref{ineq:main11}),~(\ref{jikaiyi}) and $e^{-q_i^\circ}\geq 1-q_i^\circ$, yields a lower bound of~(\ref{jikaiyi})
 \begin{align}\label{shengnan1}
 \Pr_{\textrm{CT}}\bigg[ C_0>\bar x  \bigg]\geq &\sum_{i=1}^{\infty}q_i^{\circ}(1-q_i^\circ)^{t_C}\nonumber
 \\ =&\Pr[\sigma>t_C]\sim\Gamma(\beta+1)/ \Phi(t_C)\nonumber
 \\\sim&\Gamma(\beta+1)/ \Phi(T^{\leftarrow}(\bar x)).
 \end{align}  
Next, we derive an upper bound of~(\ref{jikaiyi}). Using a similar approach to~(\ref{shengnan}) and~(\ref{ineq:76}), we obtain, for $n_1\in\mathbb{N}^+$ 
\begin{align}
\Pr_{\textrm{CT}}\left[ C_0>\bar x \right]&=\sum_{i=1}^{n_1}q_i^{\circ}e^{-q_i^{\circ}t_C}+\sum_{i=n_1}^{\infty}q_i^{\circ}e^{-q_i^{\circ}t_C}\nonumber
\\&\leq e^{-q_{n_1}^{\circ}t_C}+\sum_{i=n_1}^{\infty}q_i^{\circ}e^{-q_i^{\circ}t_C}\nonumber
\\&\lesssim \Gamma(\beta+1)/ \Phi(t_C)\nonumber
\\&\sim \Gamma(\beta+1)/ \Phi(T^{\leftarrow}(\bar x)),\nonumber
\end{align}
which, in conjunction with~(\ref{shengnan1}), yields,
\begin{align}\label{shengnan2}
 \Pr_{\textrm{CT}}\left[ C_0>\bar x  \right]\sim \Gamma(\beta+1)/ \Phi(T^{\leftarrow}(\bar x)).
\end{align}
Noting $q_i^{\circ}\sim c_\circ/i^{\alpha_\circ}$ and using Lemma~\ref{le:1.5}, we obtain, 
$\Phi(z)\sim (\alpha_\circ-1)c_\circ^{-1/\alpha_\circ}z^{1-1/\alpha}$ and $T(z)\sim (cz)^{1/\alpha_\circ}\Gamma(1-1/\alpha_\circ)$, which, together with~(\ref{shengnan2}), implies $\beta=1-1/\alpha_\circ$ and 
\begin{align}\label{shengnan4}
\Pr_{\textrm{CT}}\left[ C_0>\bar x  \right]\sim\frac{c_\circ \left(\Gamma(1-\alpha_\circ^{-1}) \right)^{\alpha_\circ}}{\alpha_\circ \bar x^{\alpha_\circ-1}}.
\end{align}
Combining~(\ref{shengnan3}) with~(\ref{shengnan4}) finishes the proof.
\end{proof}Combining Lemma~\ref{le:che} and Theorem~\ref{theo:3} yields Theorem~\ref{approx:che}.
\bibliographystyle{IEEEtran}
\bibliography{caching} 

\begin{thebibliography}{10}
\providecommand{\url}[1]{#1}
\csname url@samestyle\endcsname
\providecommand{\newblock}{\relax}
\providecommand{\bibinfo}[2]{#2}
\providecommand{\BIBentrySTDinterwordspacing}{\spaceskip=0pt\relax}
\providecommand{\BIBentryALTinterwordstretchfactor}{4}
\providecommand{\BIBentryALTinterwordspacing}{\spaceskip=\fontdimen2\font plus
\BIBentryALTinterwordstretchfactor\fontdimen3\font minus
  \fontdimen4\font\relax}
\providecommand{\BIBforeignlanguage}[2]{{%
\expandafter\ifx\csname l@#1\endcsname\relax
\typeout{** WARNING: IEEEtran.bst: No hyphenation pattern has been}%
\typeout{** loaded for the language `#1'. Using the pattern for}%
\typeout{** the default language instead.}%
\else
\language=\csname l@#1\endcsname
\fi
#2}}
\providecommand{\BIBdecl}{\relax}
\BIBdecl

\bibitem{singh2015survey}
D.~Singh and C.~K. Reddy, ``A survey on platforms for big data analytics,''
  \emph{Journal of Big Data}, vol.~2, no.~1, p.~8, 2015.

\bibitem{Karger:1997}
D.~Karger, E.~Lehman, T.~Leighton, R.~Panigrahy, M.~Levine, and D.~Lewin,
  ``Consistent hashing and random trees: Distributed caching protocols for
  relieving hot spots on the world wide web,'' in \emph{Proceedings of the
  Twenty-ninth Annual ACM Symposium on Theory of Computing}, ser. STOC '97,
  1997, pp. 654--663.

\bibitem{dynamo}
G.~DeCandia, D.~Hastorun, M.~Jampani, G.~Kakulapati, A.~Lakshman, A.~Pilchin,
  S.~Sivasubramanian, P.~Vosshall, and W.~Vogels, ``Dynamo: Amazon's highly
  available key-value store,'' in \emph{Proceedings of Twenty-first ACM SIGOPS
  Symposium on Operating Systems Principles}, ser. SOSP '07.\hskip 1em plus
  0.5em minus 0.4em\relax ACM, 2007, pp. 205--220.

\bibitem{aerospike}
``{{A}erospike},'' http://www.aerospike.com/.

\bibitem{memcached}
``{M}emcached,'' http://memcached.org/.

\bibitem{redis}
``{{R}edis},'' http://redis.io/.

\bibitem{AndrewOS}
A.~S. Tanenbaum, \emph{Modern Operating Systems}, 2nd~ed.\hskip 1em plus 0.5em
  minus 0.4em\relax Upper Saddle River, NJ, USA: Prentice Hall Press, 2001.

\bibitem{vanichpun2004output}
S.~Vanichpun and A.~M. Makowski, ``The output of a cache under the independent
  reference model: where did the locality of reference go?'' in \emph{ACM
  SIGMETRICS Performance Evaluation Review}, vol.~32, no.~1.\hskip 1em plus
  0.5em minus 0.4em\relax ACM, 2004, pp. 295--306.

\bibitem{jiantanSig}
J.~Tan, G.~Quan, K.~Ji, and N.~Shroff, ``On resource pooling and separation for
  LRU caching,'' in \emph{Proceedings of the 2018 ACM SIGMETRICS International
  Conference on Measurement and Modeling of Computer Science}.\hskip 1em plus
  0.5em minus 0.4em\relax ACM, 2018.

\bibitem{memcachNSDI}
R.~Nishtala, H.~Fugal, S.~Grimm, M.~Kwiatkowski, H.~Lee, H.~C. Li, R.~McElroy,
  M.~Paleczny, D.~Peek, P.~Saab, D.~Stafford, T.~Tung, and V.~Venkataramani,
  ``Scaling Memcache at Facebook,'' in \emph{Presented as part of the 10th
  USENIX Symposium on Networked Systems Design and Implementation (NSDI
  13)}.\hskip 1em plus 0.5em minus 0.4em\relax Lombard, IL: USENIX, 2013, pp.
  385--398.

\bibitem{cormen:2001}
T.~H. Cormen, C.~Stein, R.~L. Rivest, and C.~E. Leiserson, \emph{Introduction
  to Algorithms}, 2nd~ed.\hskip 1em plus 0.5em minus 0.4em\relax McGraw-Hill
  Higher Education, 2001.

\bibitem{Karger:1999}
D.~Karger, A.~Sherman, A.~Berkheimer, B.~Bogstad, R.~Dhanidina, K.~Iwamoto,
  B.~Kim, L.~Matkins, and Y.~Yerushalmi, ``Web caching with consistent
  hashing,'' in \emph{Proceedings of the Eighth International Conference on
  World Wide Web}, ser. WWW '99, 1999, pp. 1203--1213.

\bibitem{Wang:1999:SWC}
J.~Wang, ``A survey of web caching schemes for the Internet,'' \emph{SIGCOMM
  Computer Communication Review}, vol.~29, no.~5, pp. 36--46, Oct. 1999.

\bibitem{stoica2003chord}
I.~Stoica, R.~Morris, D.~Liben-Nowell, D.~R. Karger, M.~F. Kaashoek, F.~Dabek,
  and H.~Balakrishnan, ``Chord: a scalable peer-to-peer lookup protocol for
  Internet applications,'' \emph{IEEE/ACM Transactions on Networking (TON)},
  vol.~11, no.~1, pp. 17--32, 2003.

\bibitem{lakshman2010cassandra}
A.~Lakshman and P.~Malik, ``Cassandra: a decentralized structured storage
  system,'' \emph{ACM SIGOPS Operating Systems Review}, vol.~44, no.~2, pp.
  35--40, 2010.

\bibitem{schickinger2000simplified}
T.~Schickinger and A.~Steger, ``Simplified witness tree arguments,''
  \emph{Lecture notes in computer science}, pp. 71--87, 2000.

\bibitem{novakovic2015never}
S.~Novakovic, P.~Faraboschi, K.~Keeton, R.~Schreiber, E.~Bugnion, and
  B.~Falsafi, ``Never mind networking: Using shared non-volatile memory in
  scale-out software,'' in \emph{presentation, 2nd Int'l Workshop Rack-Scale
  Computing (WRSC 15)}, 2015.

\bibitem{edelkamp2014planning}
S.~Edelkamp, ``Planning with pattern databases,'' in \emph{Sixth European
  Conference on Planning}, 2014.

\bibitem{dietzfelbinger2009applications}
M.~Dietzfelbinger and M.~Rink, ``Applications of a splitting trick,''
  \emph{Automata, Languages and Programming}, pp. 354--365, 2009.

\bibitem{pagh2008uniform}
A.~Pagh and R.~Pagh, ``Uniform hashing in constant time and optimal space,''
  \emph{SIAM Journal on Computing}, vol.~38, no.~1, pp. 85--96, 2008.

\bibitem{drudi2015approximating}
Z.~Drudi, N.~J. Harvey, S.~Ingram, A.~Warfield, and J.~Wires, ``Approximating
  hit rate curves using streaming algorithms,'' in \emph{LIPIcs-Leibniz
  International Proceedings in Informatics}, vol.~40.\hskip 1em plus 0.5em
  minus 0.4em\relax Schloss Dagstuhl-Leibniz-Zentrum fuer Informatik, 2015.

\bibitem{gast2015transient}
N.~Gast and B.~Van~Houdt, ``Transient and steady-state regime of a family of
  list-based cache replacement algorithms,'' \emph{ACM SIGMETRICS Performance
  Evaluation Review}, vol.~43, no.~1, pp. 123--136, 2015.

\bibitem{fagin1977asymptotic}
R.~Fagin, ``Asymptotic miss ratios over independent references,'' \emph{Journal
  of Computer and System Sciences}, vol.~14, no.~2, pp. 222--250, 1977.

\bibitem{che2002}
H.~Che, Y.~Tung, and Z.~Wang, ``Hierarchical web caching systems: modeling,
  design and experimental results,'' \emph{IEEE Journal on Selected Areas in
  Communications}, vol.~20, no.~7, pp. 1305--1314, Sep 2002.

\bibitem{Jelenkovic:2004}
P.~R. Jelenkovi\'{c} and A.~Radovanovi\'{c}, ``Least-recently-used caching with
  dependent requests,'' \emph{Theoretical Computer Science}, vol. 326, no. 1-3,
  pp. 293--327, Oct. 2004.

\bibitem{Jelenkovic99asymptoticapproximation}
P.~R. Jelenkovi\'{c}, ``Asymptotic approximation of the move-to-front search
  cost distribution and least-recently-used caching fault probabilities,''
  \emph{The Annals of Applied Probability}, no.~2, pp. 430--464, 1999.

\bibitem{osogami2010fluid}
T.~Osogami, ``A fluid limit for a cache algorithm with general request
  processes,'' \emph{Advances in Applied Probability}, vol.~42, no.~3, pp.
  816--833, 2010.

\bibitem{BergerGSC14}
D.~S. Berger, P.~Gland, S.~Singla, and F.~Ciucu, ``Exact analysis of {TTL}
  cache networks,'' in \emph{32nd International symposium on Computer
  Performance, modeling, measurements, and evaluation (IFIP Performance'14)},
  Turin, Italy, October 2014, pp. 2--23.

\bibitem{garetto2016unified}
M.~Garetto, E.~Leonardi, and V.~Martina, ``A unified approach to the
  performance analysis of caching systems,'' \emph{ACM Transactions on Modeling
  and Performance Evaluation of Computing Systems}, vol.~1, no.~3, p.~12, 2016.

\bibitem{Fricker:2012}
C.~Fricker, P.~Robert, and J.~Roberts, ``A versatile and accurate approximation
  for {LRU} cache performance,'' in \emph{Proceedings of the 24th International
  Teletraffic Congress}, ser. ITC '12, 2012, pp. 8:1--8:8.

\bibitem{roberts2013exploring}
J.~Roberts and N.~Sbihi, ``Exploring the memory-bandwidth tradeoff in an
  information-centric network,'' in \emph{Teletraffic Congress (ITC), 2013 25th
  International}.\hskip 1em plus 0.5em minus 0.4em\relax IEEE, 2013, pp. 1--9.

\bibitem{shanmugam2013femtocaching}
K.~Shanmugam, N.~Golrezaei, A.~G. Dimakis, A.~F. Molisch, and G.~Caire,
  ``Femtocaching: Wireless content delivery through distributed caching
  helpers,'' \emph{IEEE Transactions on Information Theory}, vol.~59, no.~12,
  pp. 8402--8413, 2013.

\bibitem{ioannidis2016adaptive}
S.~Ioannidis and E.~Yeh, ``Adaptive caching networks with optimality
  guarantees,'' in \emph{Proceedings of the 2016 ACM SIGMETRICS International
  Conference on Measurement and Modeling of Computer Science}.\hskip 1em plus
  0.5em minus 0.4em\relax ACM, 2016, pp. 113--124.

\bibitem{berger2014exact}
D.~S. Berger, P.~Gland, S.~Singla, and F.~Ciucu, ``Exact analysis of TTL cache
  networks,'' \emph{Performance Evaluation}, vol.~79, pp. 2--23, 2014.

\bibitem{Choungmo2014}
N.~C. Fofack, P.~Nain, G.~Neglia, and D.~Towsley, ``Performance evaluation of
  hierarchical {TTL}-based cache networks,'' \emph{Computer Networks}, vol.~65,
  pp. 212 -- 231, 2014.

\bibitem{Carter:1977}
J.~L. Carter and M.~N. Wegman, ``Universal classes of hash functions,'' in
  \emph{Proceedings of the Ninth Annual ACM Symposium on Theory of Computing},
  ser. STOC '77.\hskip 1em plus 0.5em minus 0.4em\relax New York, NY, USA: ACM,
  1977, pp. 106--112.

\bibitem{Wegman79}
M.~N. Wegman and J.~L. Carter, ``New classes and applications of hash
  functions,'' in \emph{20th Annual Symposium on Foundations of Computer
  Science (sfcs 1979)}, Oct 1979, pp. 175--182.

\bibitem{Mitzenmacher:2008}
M.~Mitzenmacher and S.~Vadhan, ``Why simple hash functions work: Exploiting the
  entropy in a data stream,'' in \emph{Proceedings of the Nineteenth Annual
  ACM-SIAM Symposium on Discrete Algorithms}, ser. SODA '08.\hskip 1em plus
  0.5em minus 0.4em\relax Philadelphia, PA, USA: Society for Industrial and
  Applied Mathematics, 2008, pp. 746--755.

\bibitem{lee99}
L.~Breslau, P.~Cao, L.~Fan, G.~Phillips, and S.~Shenker, ``Web caching and
  Zipf-like distributions: evidence and implications,'' in \emph{Proceedings of
  the 18th Conference on Information Communications}, 1999.

\bibitem{fricker2012impact}
C.~Fricker, P.~Robert, J.~Roberts, and N.~Sbihi, ``Impact of traffic mix on
  caching performance in a content-centric network,'' in \emph{Computer
  Communications Workshops (INFOCOM WKSHPS), 2012 IEEE Conference on}.\hskip
  1em plus 0.5em minus 0.4em\relax IEEE, 2012, pp. 310--315.

\bibitem{cha2007tube}
M.~Cha, H.~Kwak, P.~Rodriguez, Y.-Y. Ahn, and S.~Moon, ``I tube, you tube,
  everybody tubes: analyzing the world's largest user generated content video
  system,'' in \emph{Proceedings of the 7th ACM SIGCOMM conference on Internet
  measurement}.\hskip 1em plus 0.5em minus 0.4em\relax ACM, 2007, pp. 1--14.

\bibitem{jelenkovic2007lru}
P.~R. Jelenkovi{\'c} and X.~Kang, ``LRU caching with moderately heavy request
  distributions,'' in \emph{2007 Proceedings of the Fourth Workshop on Analytic
  Algorithmics and Combinatorics (ANALCO)}.\hskip 1em plus 0.5em minus
  0.4em\relax SIAM, 2007, pp. 212--222.

\bibitem{guocongINFO}
G.~Quan, K.~Ji, and J.~Tan, ``LRU caching with dependent competing requests,''
  in \emph{2018 IEEE Conference on Computer Communications (INFOCOM)},
  Honolulu, HI, April 2018.

\bibitem{berthet2017approximation}
C.~Berthet, ``Approximation of LRU caches miss rate: Application to power-law
  popularities,'' \emph{arXiv preprint arXiv:1705.10738}, 2017.

\bibitem{fill:1996}
J.~Fill, ``An exact formula for the move-to-front rule for self-organizing
  lists,'' \emph{Journal of Theoretical Probability}, vol.~9, no.~1, pp.
  113--160, 1996.

\bibitem{regularVariation}
N.~H. Bingham, C.~M. Goldie, and J.~L. Teugels, \emph{Regular Variation}.\hskip
  1em plus 0.5em minus 0.4em\relax Cambridge University Press, 1987.

\bibitem{chung2006complex}
F.~R. Chung and L.~Lu, \emph{Complex graphs and networks}.\hskip 1em plus 0.5em
  minus 0.4em\relax American mathematical society Providence, 2006, vol. 107.

\end{thebibliography}

\end{document}